\documentclass[a4paper,english]{lipics-v2021}

\title{Finding a Maximum Clique in a Disk Graph} 

\author{Jared Espenant}{Department of Computer Science, University of Saskatchewan, Saskatoon, Saskatchewan, Canada}{jae608@usask.ca}{https://orcid.org/0000-0002-1825-0097}{}

\author{J. Mark Keil}{Department of Computer Science, University of Saskatchewan, Saskatoon, Saskatchewan, Canada}{keil@cs.usask.ca}{}{}

\author{Debajyoti Mondal\footnote{Corresponding author}}{Department of Computer Science, University of Saskatchewan, Saskatoon, Saskatchewan, Canada}{dmondal@cs.usask.ca}{https://orcid.org/0000-0002-7370-8697}{}

\authorrunning{J. Espenant and J.\,M. Keil and D. Mondal} 

\Copyright{Jared Espenant and J. Mark Keil and Debajyoti Mondal} 



\relatedversion{} 

\funding{The work is supported in part by the Natural Sciences and Engineering Research Council of Canada (NSERC).}

\nolinenumbers 

 \EventLongTitle{Preliminary results appeared at the 39th International Symposium on Computational Geometry (SoCG 2023)}
\ArticleNo{32}

\begin{document}

\maketitle

\begin{abstract}
A disk graph is an intersection graph of disks in the Euclidean plane, where the disks correspond to the vertices of the graph and a pair of vertices are adjacent if and only if their corresponding disks intersect. The problem of determining the time complexity of computing a maximum clique in a disk graph is a long-standing open question that has been very well studied in the literature. The problem is known to be open even when the radii of all the disks are in the interval $[1,(1+\varepsilon)]$, where $\varepsilon>0$. If all the disks are unit disks then there exists an $O(n^3\log n)$-time algorithm to compute a maximum clique, which is the best-known running time for over a decade. Although the problem of computing a maximum clique in a disk graph remains open, it is known to be APX-hard for the intersection graphs of many other convex objects such as intersection graphs of ellipses, triangles, and  a combination of unit disks and axis-parallel rectangles.  Here we obtain the following results. 

\begin{enumerate}[-]
    \item We give an algorithm to compute a  maximum clique in a unit disk graph in $O(n^{2.5}\log n)$-time, which improves the previously best known running time of  $O(n^3\log n)$ [Eppstein '09].
    
    \item  
    We extend a widely used `co-2-subdivision approach' to prove that computing a maximum clique in a combination of unit disks and axis-parallel rectangles is NP-hard to approximate within $4448/4449 \approx 0.9997 $. The use of a `co-2-subdivision approach' was previously thought to be unlikely in this setting [Bonnet et al. '20]. Our result improves the previously known inapproximability factor of $7633010347/7633010348\approx 0.9999$.
    
    \item We show that the parameter minimum lens width of the disk arrangement may be used to make progress  in the case when disk radii are in $[1,(1+\varepsilon)]$. For example, if the minimum lens width is at least  $0.265$ and $ \varepsilon\le 0.0001$, which still allows for non-Helly triples in the arrangement, then  one can find a maximum clique  in polynomial time.
    
\end{enumerate}
\end{abstract}

\section{Introduction}

An \emph{intersection graph} of a set $S$ of geometric objects is a graph where each object in $S$ corresponds to a vertex in $G$ and two vertices in $G$ are adjacent if and only if the corresponding objects intersect. A set of vertices $C\subseteq V$ is called a \emph{clique} if they are mutually adjacent.  In this paper, we are interested in the problem of finding a \emph{maximum clique}, i.e., a largest set of mutually adjacent vertices. We mainly focus on \emph{disk graphs}, i.e.,  the intersection graphs of disks in $\mathbb{R}^2$ (Figure~\ref{fig:intro}).  Disk graphs are often used to model ad-hoc wireless  networks~\cite{huson1995broadcast}. 

The time complexity question for finding a maximum clique in a disk graph is known to be open for over two   decades~\cite{ambuhl2005clique,BRS06,DBLP:conf/waoa/Fishkin03}. The question is open even in severely restricted setting such as when the radii of the disks are of two types~\cite{cabopen} or when the disk radii are in the interval $[1,1+\varepsilon]$ for a fixed 
 $\varepsilon>0$~\cite{DBLP:journals/jacm/BonamyBBCGKRST21}. However, there exists randomized EPTAS, deterministic PTAS, and  subexponential-time algorithms  for computing a maximum clique in arbitrary disk graphs~\cite{DBLP:journals/jacm/BonamyBBCGKRST21,DBLP:conf/compgeom/BonnetG0RS18}. For unit disk graphs, i.e., when all  the radii are the same, Clark et al.~\cite{DBLP:journals/dm/ClarkCJ90} showed that a maximum clique can be found in $O(n^{4.5})$-time. Their algorithm searches for a maximum clique over all the lenses of pairwise intersecting disks. Later, Eppstein~\cite{DBLP:conf/wg/Eppstein09} showed how the algorithm could be implemented in $O(n^3\log n)$-time by searching through a careful ordering of the lenses and  using a data structure of~\cite{DBLP:journals/jal/AggarwalIKS91} to maintain a maximum clique throughout the search. Faster algorithms are known in constrained settings where the centers of the disks lie within a narrow horizontal strip~\cite{breu}. Polynomial-time algorithms exist for  many other intersection graph classes such as for circle graphs~\cite{DBLP:journals/algorithmica/Tiskin15},  trapezoid graphs~\cite{DBLP:journals/dam/FelsnerMW97}, circle trapezoid graphs~\cite{DBLP:journals/dam/FelsnerMW97}, intersection graphs of axis-parallel rectangles~\cite{DBLP:journals/jal/ImaiA83}, and so on. 

\begin{figure}[pt]
  \centering
  \includegraphics[width=0.80\textwidth]{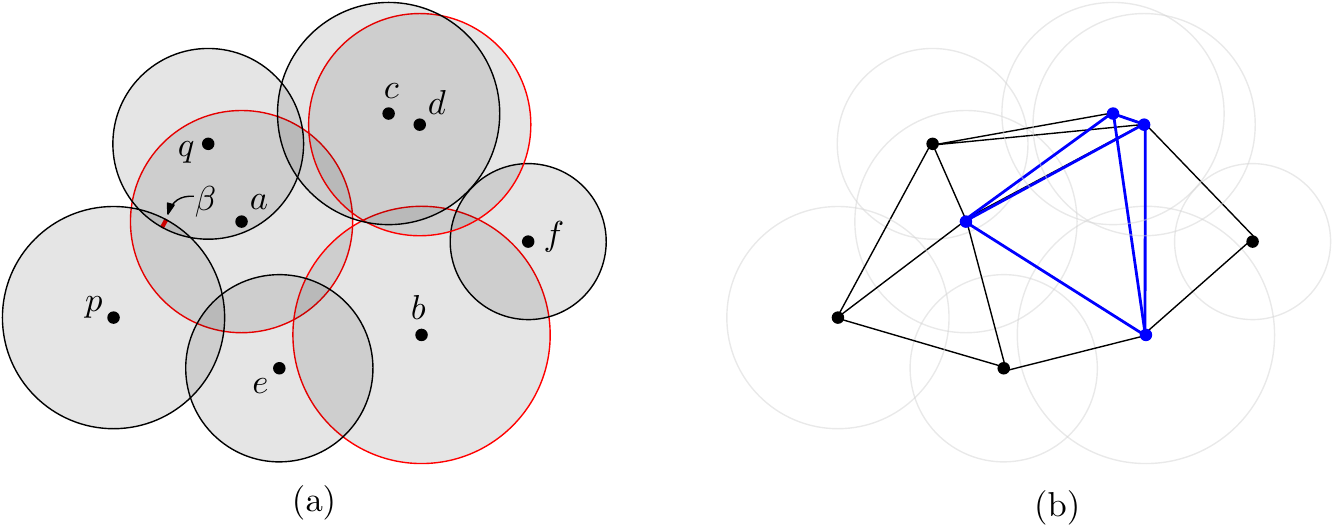}
\caption{(a) An arrangement of disks $\mathcal{A}$. The minimum width $\beta$ over all the lenses is determined by the disks centered at $p$ and $q$. A non-Helly triple is shown in red. (b) A disk graph corresponding to $\mathcal{A}$, where a maximum clique is shown in blue. 
}\label{fig:intro}
\end{figure}

Although the maximum clique problem is open for disk graphs, a number of APX-hardness results are known in the literature for intersection graphs of other geometric objects. A common approach to prove the NP-hardness result for computing a maximum clique in an intersection graph class $\mathcal{I}$ is to take a \emph{co-$2k$-subdivision approach}, as follows. A \emph{$2k$-subdivision}, where $k$ is a positive integer, of a graph $G$ is obtained by replacing each edge $(u,v)$ of $G$ with a path $(u,d_1,\ldots, d_{2k},v)$ of $2k$ division vertices. A  co-$2k$-subdivision approach takes a graph class for which finding a maximum independent set is NP-hard and shows that the complement graph of its $2k$-subdivision has an intersection representation in class $\mathcal{I}$. Since the NP-hardness of computing a maximum independent set is preserved by the even subdivision~\cite{ChlebikC07} and since a maximum independent set in a graph corresponds to a maximum clique in the complement graph, this establishes the NP-hardness result for computing a maximum clique in class $\mathcal{I}$. Some of the intersection graph classes for which the maximum clique problem has been proved to be APX-hard using the co-$2k$-subdivision approach are intersection graphs of ellipses~\cite{ambuhl2005clique}, triangles~\cite{ambuhl2005clique}, string graphs~\cite{DBLP:journals/dm/MiddendorfP92}, grounded string graphs~\cite{DBLP:journals/corr/abs-2107-05198}, and so on. Cabello~\cite{DBLP:journals/dcg/CabelloCL13} used the co-$2k$-subdivision approach to prove the NP-hardness of computing a maximum clique   in the intersection graph of rays, which settled a 21-year-old open problem posed by Kratochv\'il and Ne\v{s}et\v{r}il~\cite{kratochvil1990independent}. To the best of our knowledge, no hardness of approximation result is known for this graph class. 

Since there is  strong evidence that a co-$2k$-subdivision approach may not be sufficient to prove the NP-hardness of computing a maximum clique in a disk graph~\cite{DBLP:conf/compgeom/BonnetG0RS18}, Bonnet et al.~\cite{bonnet_et_al:LIPIcs:2020:13258} attempted to explore alternative approaches. While they were not able to prove the NP-hardness for disk graphs, they showed that the problem of computing a maximum clique in an intersection graph that contains both unit disks and axis-parallel rectangles  is not approximable within a factor of 7633010347/7633010348 in polynomial time, unless P=NP. This result is interesting since the maximum clique problem is polynomial-time solvable when all objects are either unit disks~\cite{DBLP:journals/dm/ClarkCJ90} or  axis-parallel rectangles~\cite{DBLP:journals/jal/ImaiA83}. To obtain this result, Bonnet et al.~\cite{bonnet_et_al:LIPIcs:2020:13258}  introduced a new problem called `Max Interval Permutation Avoidance', proved it to be APX-hard,  and reduced it to the problem of computing a maximum clique in a combination of unit disks and axis-parallel rectangles. Furthermore, they stated that  the intersection graph of unit disks and axis-parallel rectangles is ``a class for which the co-2-subdivision approach does not seem to work''.



\subsection*{Our Contribution} In this paper we make significant progress on the maximum clique problem for unit   disk graphs,   disk graphs with disk radii lying in the interval $[1,1+\varepsilon]$, and  intersection graphs of unit disks and axis-parallel  rectangles.

\textbf{Unit disk graph:} We give an algorithm to compute a  maximum clique in a unit disk graph in $O(n^{2.5}\log n)$-time, which improves the previously best known running time of  $O(n^3\log n)$~\cite{DBLP:conf/wg/Eppstein09}. Our algorithm is based on a divide-and-conquer approach that, unlike the previous algorithms that search a clique over all the lenses, shows  how to efficiently merge solutions to the subproblems to achieve a faster time complexity. Such techniques have  previously been used to accelerate computation for other computational geometry problems, e.g., when finding a closest pair in a point set~\cite{DBLP:conf/focs/ShamosH75}, but appeared to be highly non-trivial while adapting it for  the unit disk graph setting.
    
\textbf{Intersection graph of unit disks and axis-parallel rectangles:} We extend the co-2-subdivision approach to prove a $(4448/4449 \approx 0.9997)$-inapproximability result for computing a maximum clique in an intersection graph that contains both   unit disks and axis-parallel rectangles, and thus improve the previously known inapproximability factor of $7633010347/7633010348\approx 0.9999 $~\cite{bonnet_et_al:LIPIcs:2020:13258}. Note that the use of a  co-2-subdivision approach  was previously thought to be unlikely in this setting by Bonnet et al.~\cite{bonnet_et_al:LIPIcs:2020:13258}. The key idea behind our NP-hardness reduction is to show that every Hamiltonian cubic graph admits a well-behaved edge orientation and vertex labeling, i.e., its vertices can be labeled and the edges can be oriented such that every vertex has two outgoing or two incoming edges where the labels of these corresponding neighbors are consecutive. While such orientation and labeling are of independent interest, they allow us to represent the complement of the 2-subdivision of a Hamiltonian cubic graph using a combination of unit disks and axis-parallel rectangles.
    
\textbf{$(\varepsilon,\beta)$-disk graph:} In an attempt to make progress on the case when the disk radii are in the interval $[1,1+\varepsilon]$, we introduce 
$(\varepsilon,\beta)$-disk graphs. A \emph{$(\varepsilon,\beta)$-disk graph}, where $\varepsilon$ and $\beta$ are positive constants, is a disk graph where the radii of the disks are in the interval $[1,1+\varepsilon]$ and every lens is of width at least $\beta$. The parameter $\beta$ can be thought of as the minimum width over all the lenses in the disk arrangement, where a \emph{lens} is the convex intersection region of a pair of disks (Figure~\ref{fig:intro}(a)). We show that the parameter $\beta$, i.e., the minimum lens width of the disk arrangement, may be used to make progress  in the case when disk radii are in $[1,(1+\varepsilon)]$. 
For example, if the minimum lens width is at least  $0.265$,  then one can find a maximum clique for  $ \varepsilon\le 0.0001$ in polynomial time.

The existence of \emph{non-Helly triple} in a disk arrangement, i.e., three pairwise intersecting disks without any common point of intersection (Figure~\ref{fig:intro}(a)), typically makes the problem of finding a clique challenging {\color{black}(see Sec.~\ref{sec:nonht})}. Since $\beta$ is a lower bound on the width of every lens, a natural question is whether our choice for $\beta\ge 0.265$ already forbids the existence of non-Helly triples. We note that our choice for $\beta$ still allows for non-Helly triples, and thus the result is non-trivial. We show that the lower bound on $\beta$ could be leveraged to find for each non-Helly triple, a maximum clique that includes this triple. This extends the prior approach of finding a maximum clique in a unit disk graph that searches over all the pairwise intersecting disks~\cite{DBLP:journals/dm/ClarkCJ90} to a more general setting where the disk radii are in $[1,1.0001]$. We believe that our proposed  approach is interesting from the perspective of finding a way to make progress beyond unit disks even though the lower bound on $\beta$ is large and the gain on $\varepsilon$ is small.   
 
\section{Preliminaries}

By $D^r_q$ we denote a disk with radius $r$ and center $q$. For the simplicity of the presentation, sometimes we omit the radius and simply use $D_q$ to denote a disk with center $q$. Let $D_p$ and $D_q$ be a pair of disks. By $L(D_p,D_q)$ we denote the lens  (i.e. the intersection region) of these disks (Figure~\ref{fig:lens}(a)). For a line segment $ab$, we denote by $|ab|$ the length of the segment or the Euclidean distance between the points $a$ and $b$. The \emph{width} of a lens $L(D_p,D_q)$ is the length of the line segment determined by the intersection of $pq$ and $L(D_p,D_q)$.

Let $G$ be a graph. The \emph{complement graph} $ \overline{G}$ of $G$ is a graph on the same set of vertices where   $\overline{G}$ contains an edge if and only if it does not appear in $G$.  A set $S$ of vertices in $G$ is called  \emph{independent} if no two vertices in $S$ are adjacent in $G$. A \emph{maximum independent set $\alpha(G)$} is an independent set of largest cardinality. $G$ is called \emph{bipartite} if its vertices can be partitioned into two independent sets. $G$ is called a \emph{cobipartite graph} if the  complement graph of $G$ is a bipartite graph. $G$ is called  \emph{cubic} if every vertex of $G$ is of degree three. $G$ is \emph{Hamiltonian} if it has  a cycle that contains each vertex of $G$ exactly once. 


\begin{figure}[pt]
    \centering
    \includegraphics[width=0.80\textwidth]{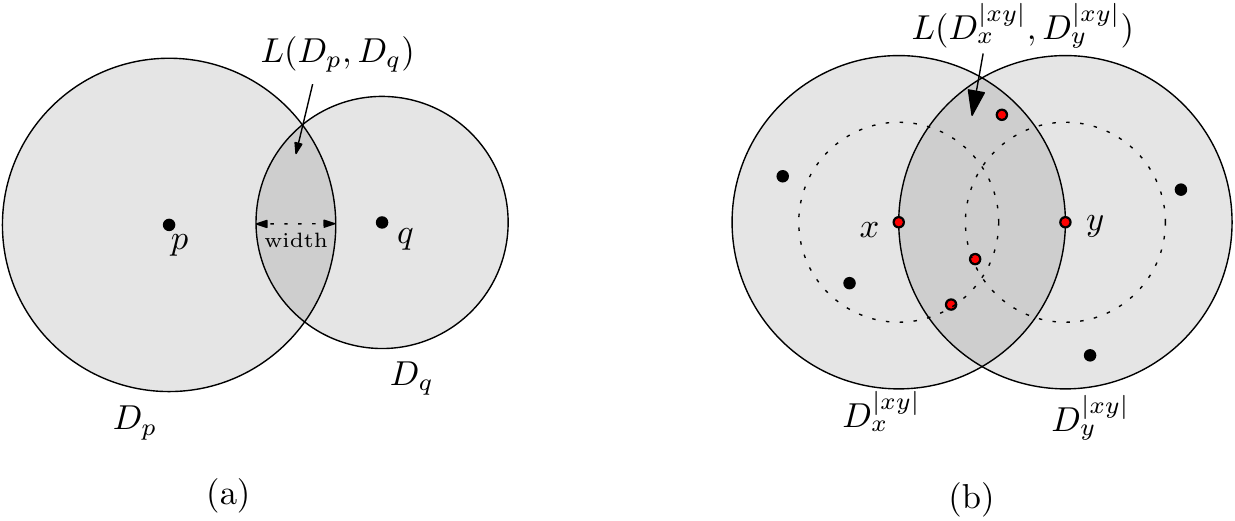}
    \caption{(a) Illustration for a lens. (b) Computation of a maximum clique, where the centers of the disks are shown in dots. The unit disks centered at $x$ and $y$ are shown in dotted circles. The centers inside $L(D^{|xy|}_x,D^{|xy|}_y)$ are shown in red. }
    \label{fig:lens}
\end{figure}
 
\section{Unit Disk Graph (UDG)}

In this section we provide an $O(n^{2.5}\log n)$-time algorithm to compute a maximum clique in a unit disk graph, where a geometric representation of the graph is given as an input. 

Clark et al.~\cite{DBLP:journals/dm/ClarkCJ90} gave an $O(n^{4.5})$-time algorithm to compute a maximum clique in a unit disk graph. The idea of the algorithm is  as follows. For each edge $(x,y)$ of the graph, consider two disks $D_x$ and $D_y$ such that their boundaries  pass through $y$ and $x$, respectively. Let $S$ be the set of unit disks with centers in $L(D^{|xy|}_x,D^{|xy|}_y)$. Clark et al. showed that the subgraph of $G$ induced by the vertices corresponding to $S$ is a cobipartite graph $G(S)$. One can thus find a maximum clique in $G(S)$ by computing a  maximum 
independent set in the bipartite graph $\overline{G(S)}$. If $(x,y)$ is the longest edge of a maximum clique $M$ in $G$, then $S$ must include all the centers of the disks in $M$ and $G(S)$ will contain the largest clique in $G$. Therefore, one can try the above strategy over all edges and find a maximum clique in $G$. Since $G$ contains $O(n^2)$ edges and since  a maximum independent set in a bipartite graph can be computed in $O(n^{2.5})$ time by leveraging a maximum  matching~\cite{DBLP:journals/siamcomp/HopcroftK73}, the running time becomes  $O(n^{4.5})$.

Breu~\cite{breu} observed that  Clark et al.'s approach~\cite{DBLP:journals/dm/ClarkCJ90} to find a   maximum clique can be implemented in $O(n^{3.5}\log n)$ time using a result of Aggarwal et al.~\cite{DBLP:journals/jal/AggarwalIKS91}. Specifically, Aggarwal et al.~\cite{DBLP:journals/jal/AggarwalIKS91}    showed how to compute a {\color{black}maximum independent set }  in $\overline{G(S)}$ in $O(n^{1.5} \log n)$ time  using a data structure of~\cite{DBLP:journals/jal/HershbergerS91,DBLP:journals/siamcomp/ImaiA86}, and hence over $O(n^2)$ lenses the running time becomes $O(n^{3.5}\log n)$. 

Eppstein~\cite{DBLP:conf/wg/Eppstein09} observed that while searching through the lenses, instead of computing the maximum independent set from scratch, one can exploit geometric properties to efficiently update and maintain a maximum independent set as follows. For a unit disk center $p$, let $q$ be a point on the plane  such that $|pq|=2$. Consider two disks $D_p$ and $D_q$ such that their boundaries  pass through $q$ and $p$, respectively. One can now rotate the lens $L(D^{|pq|}_p,D^{|pq|}_q)$ around $p$ and update the maximum independent set in the graph corresponding to $L(D^{|pq|}_p,D^{|pq|}_q)$ each time a point  (i.e., a center of a unit disk) enters or exists from the lens. An update can be processed by an alternating path search in $O(n\log n)$ time~\cite{DBLP:journals/jal/AggarwalIKS91}. Since the number of changes to $L(D^{|pq|}_p,D^{|pq|}_q)$ is bounded by $O(n)$, the time spent for $p$ is $O(n^2 \log n)$. Hence the overall running time is $O(n^3 \log n)$.



\subsection{Idea of Our Algorithm} \label{sec:idea}
Let $G$ be a disk graph with $n$ vertices, where each disk is of radius $r$. Let $P$ be the set of centers of the disks corresponding to the vertices of $G$.  To find a maximum clique we take  a divide-and-conquer approach as follows.

We rotate the plane so that no two points are in the same vertical or horizontal line. It is straightforward to perform  such a rotation in $O(n^2)$ time.  
 We sort the points in $P$ with respect to their x-coordinates and 
  find a vertical line $V$ through a median x-coordinate such that at most $\lceil n/2 \rceil$ points of $P$ are on each half-plane of $V$. Let $P_l$ and $P_r$ be the points on the closed left halfplane and closed right halfplane of $V$, respectively.  We will find a maximum clique in $P_l$ and $P_r$ recursively.
  
Let $M$ be a maximum clique in $G$. If the set of disk centers corresponding to $M$ is a subset of either $P_l$ or $P_r$, then such a clique must be returned as a solution to one of these two subproblems.  Otherwise, each of $P_l$ and $P_r$ contains some points of  $M$. To tackle such a case, it suffices to find a maximum clique in the vertical slab  between the vertical lines $V_l$ and $V_r$, where $V_l$ and $V_r$ are $2r$ units apart from $V$ on the left halfplane and right halfplane, respectively.  Let $Q\subseteq P$ be the set of points in the vertical slab.  Then the maximum clique of $G$ is the maximum clique found over the disks corresponding to the sets $P_l$, $P_r$ and $Q$.  
  
Let $T(n)$ be the time to compute a  maximum clique in $G$. Let $F(n)$ be the time to compute a maximum clique in the vertical slab. Then {\color{black}$T(n)$} is defined as follows.
  \begin{align}\label{eq1}\centering
   T(n) = 2T\left(\dfrac{n}{2}\right) + F(n).
  \end{align}
   
We now sort the points   of $Q$ with respect to their y-coordinates and find a horizontal line $H$ through the median y-coordinate such that at most $\lceil |Q|/2 \rceil$ points of $Q$ are on each half-plane of $H$. Let $Q_t$ and $Q_b$ be the points on the closed top halfplane and closed bottom halfplane of $H$, respectively. We now find a maximum clique in $Q_t$ and $Q_b$ recursively.  If the set of disk centers corresponding to $M$ is a subset of either $Q_t$ or $Q_b$, then such a clique must be returned as a solution to one of these two subproblems. Otherwise, each of $Q_t$ and $Q_b$ contains some points of  $M$. 
It now suffices to find a maximum clique in the square $S$ of side length $4r$ with its center located at the intersection point of $V$ and $H$ (Figure~\ref{triv}(a)). Let $B(n)$ be the time to compute the maximum clique in $S$. 
Then $F(n)$ is defined as follows.
\begin{align}\label{eq2}\centering
   F(n) = 2F\left(\dfrac{n}{2}\right) + B(n).
\end{align}

In the following, we will show that a maximum clique in $S$ can be computed in $O(n^{2.5}\log n)$ time. Consequently, $B(n)\in  O(n^{2.5}\log n)$ and by Equation~\ref{eq2} and master theorem, $F(n) \in  O(n^{2.5}\log n)$. Consequently, the   time complexity determined by Equation~\ref{eq1} is $O(n^{2.5}\log n)$.  {\color{black}Note that computing a maximum clique in the square $S$ of side length $4r$ appears to be the bottleneck of our algorithm.}

\subsection{Computing a Maximum Clique in the Square $S$ } 

Let $M$ be a maximum clique in $G$ and let $C$ be the centers of the disks in $M$. Assume that $C\not\subseteq P_l$, $C\not\subseteq P_r$, $C\not\subseteq Q_t$ and $C\not\subseteq Q_b$, i.e., $C$ is a subset of the points  in $S$. We now show how to find $M$. 
%
Let $o$ be the center of $S$, i.e., the intersection point  of  $H$ and $V$. Without loss of generality assume that $o$  is at $(0,0)$. Let $R_i$, where $1\le i\le 4$, be the region  determined by the intersection of the $i$th quadrant and $S$ (Figure~\ref{triv}(a)). We now give two remarks. Remark~\ref{rem:opp} follows directly from our assumption that $C\not\subseteq P_l$, $C\not\subseteq P_r$, $C\not\subseteq Q_t$ and $C\not\subseteq Q_b$.
\begin{remark}\label{rem:opp}
$C$ must satisfy at least one of the following two conditions.  
(a) $R_1$ and $R_3$  each contains a point from $C$.  
(b) $R_2$ and $R_4$  each contains a point from $C$.
\end{remark} 

 
 \begin{figure} 
     \centering
     \includegraphics[width=\textwidth]{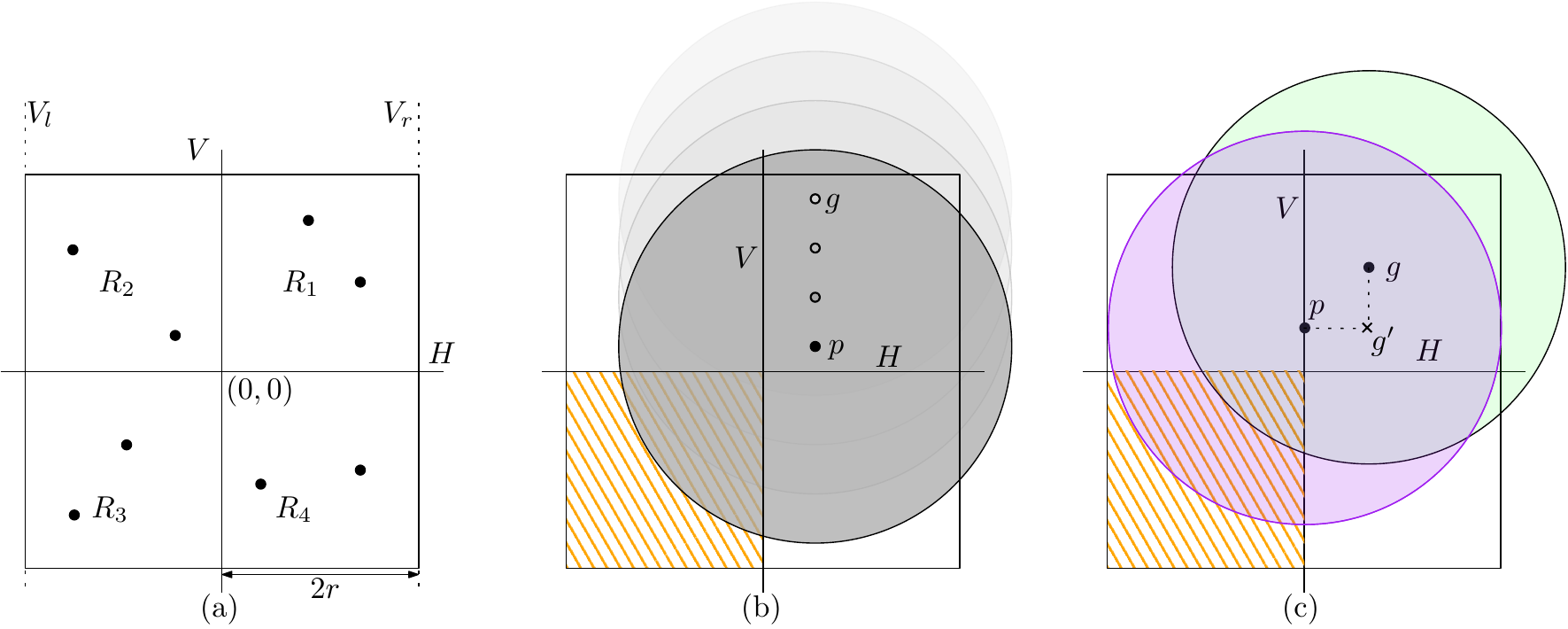}
    \caption{   (a) The square $S$ with the disk centers   in black dots. (b)--(c) Illustration for Remark~\ref{rem:triv}.}
    \label{triv}
 \end{figure}
 
\begin{remark}\label{rem:triv}
Let $p$ and $g$ be two points inside $R_1$ where the x- and y-coordinates of $g$ are at least as large as that of $p$. Let $D^{2r}_p$ and $D^{2r}_g$ be two disks of the same radius $2r$ centered at $p$ and $g$, respectively. Then $(R_3\cap D^{2r}_g) \subseteq (R_3\cap D^{2r}_p$). 
\end{remark}  
\begin{proof}
Consider first the case when $g$ and $p$ lie on the same vertical line (Figure~\ref{triv}(b)). Note that  the interval $(H\cap D^{2r}_g)$ increases as we move the center $g$ vertically downward and the interval reaches the maximum when $g$ hits $H$. Therefore,  $(H\cap D^{2r}_g) \subseteq (H\cap D^{2r}_p)$. Since $g$ is vertically above $p$ and both have the same radius, $(R_3\cap D^{2r}_g) \subseteq (R_3\cap D^{2r}_p$). The argument when $g$ and $p$ lie on the same horizontal  line is symmetric.

Consider now the case when x- and y- coordinates of $g$ are larger than that of $p$ (Figure~\ref{triv}(c)). We can move $D^{2r}_g$ vertically down to reach a point $g'$ that has the same y-coordinate as that of $p$. Consequently, $(R_3\cap D^{2r}_g) \subseteq (R_3\cap D^{2r}_{g'})$. Finally, we move $D^{2r}_{g'}$ towards $p$. Hence we obtain $(R_3\cap D^{2r}_{g'}) \subseteq (R_3\cap D^{2r}_p)$. 
\end{proof}


We are now ready to describe the algorithm. 
 The algorithm   considers two cases depending on whether every disk center in $C$ is within a distance of $2r$ from $o$.  It  processes each case in $O(n^{2.5}\log n)$ time, and then returns the maximum clique found over the whole process.


The high-level idea for finding a maximum clique is as follows. For the first case, we assume every disk center in $C$ to be  within a distance of $2r$ from $o$. The algorithm makes a guess for the farthest disk center $q$ in $C$ from $o$ and then finds the other disks in the maximum clique by defining a lens that would contain all the disk centers of $C$. For the second case, we assume that at least one disk center in $C$ has a distance of more than $2r$ from $o$. The algorithm makes a guess for the first point $p\in C$ in some particular point ordering and then finds the other disks in the maximum clique by defining a lens that would contain all the disk centers of $C$. We now describe the details.




\subsubsection*{Case 1 (Every disk center in $C$ is within a distance of $2r$ from $o$)} 

Let $q$ be a point of $C$ that has the largest distance from $o$. Without loss of generality assume that $q$ lies in $R_2$. We now order the points of $R_2$ that are within distance $2r$ from $o$  in decreasing order of their distances from $o$ (breaking ties arbitrarily). Figure~\ref{order}(a) illustrates this order in orange concentric  circles.  Let $\sigma$ be the resulting point ordering. We iteratively consider each point in $\sigma$ to be $q$ and then find  a maximum clique as follows.

Let $D^{|oq|}_o$ be a disk centered at $o=(0,0)$ such that its boundary passes through $q$ (Figure~\ref{order}(b)). Since $q$ is the furthest point of $C$ from $o$, every point of $C$ is contained in $D^{|oq|}_o$. Let $t$ be a point in $R_4$ that lies on the line through $o$ and $q$ at a distance of $2r$ from $q$.  
 
 We now show that every point of $C$ is in $L(D^{2r}_q,D^{2r}_t)$. Suppose for a contradiction that there exists a point $e\in C$ which is not in $L(D^{2r}_q,D^{2r}_t)$. If $e$ belongs to $S \setminus D^{2r}_q$, then $D_e$ cannot intersect $D_q$. Therefore, $e$ must lie in the region $D^{|oq|}_o\cap (D^{2r}_q\setminus D^{2r}_t)$. Note that  $q,o$ and $t$ lie on the same line. Since the boundaries of both $D^{|oq|}_o$ and $D^{2r}_t$ pass through $q$ and $|qt|\ge |qo|$, we have $D^{|oq|}_o\subseteq D^{2r}_t$, and hence the point $e$ cannot exist. 
 
 \begin{figure} 
     \centering
     \includegraphics[width=\textwidth]{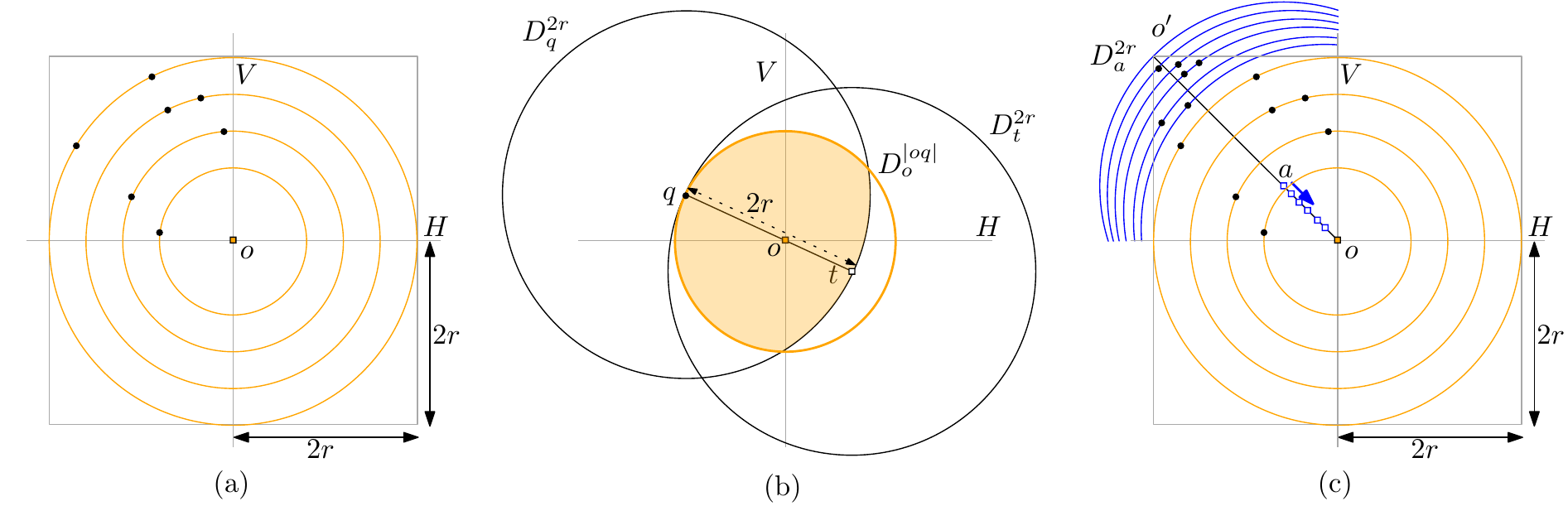}
    \caption{{\color{black}(a) Illustration for point ordering in $R_2$ in Case 1. (b) Illustration for $D^{2r}_q$  and $D^{2r}_t$. (c) Point ordering  in Case 2.}}
    \label{order}
 \end{figure}
 
Since the intersection graph induced by the disks with centers in $L(D^{2r}_q,D^{2r}_t)$ is a  cobipartite graph~\cite{DBLP:journals/dm/ClarkCJ90},  a maximum clique in this graph can be computed in $O(n^{1.5}\log n)$ time~\cite{DBLP:journals/jal/AggarwalIKS91}. Over all choices for $q$ in $\sigma$, the running time becomes $O(n^{2.5}\log n)$. 


 \begin{figure} 
     \centering
     \includegraphics[width=\textwidth]{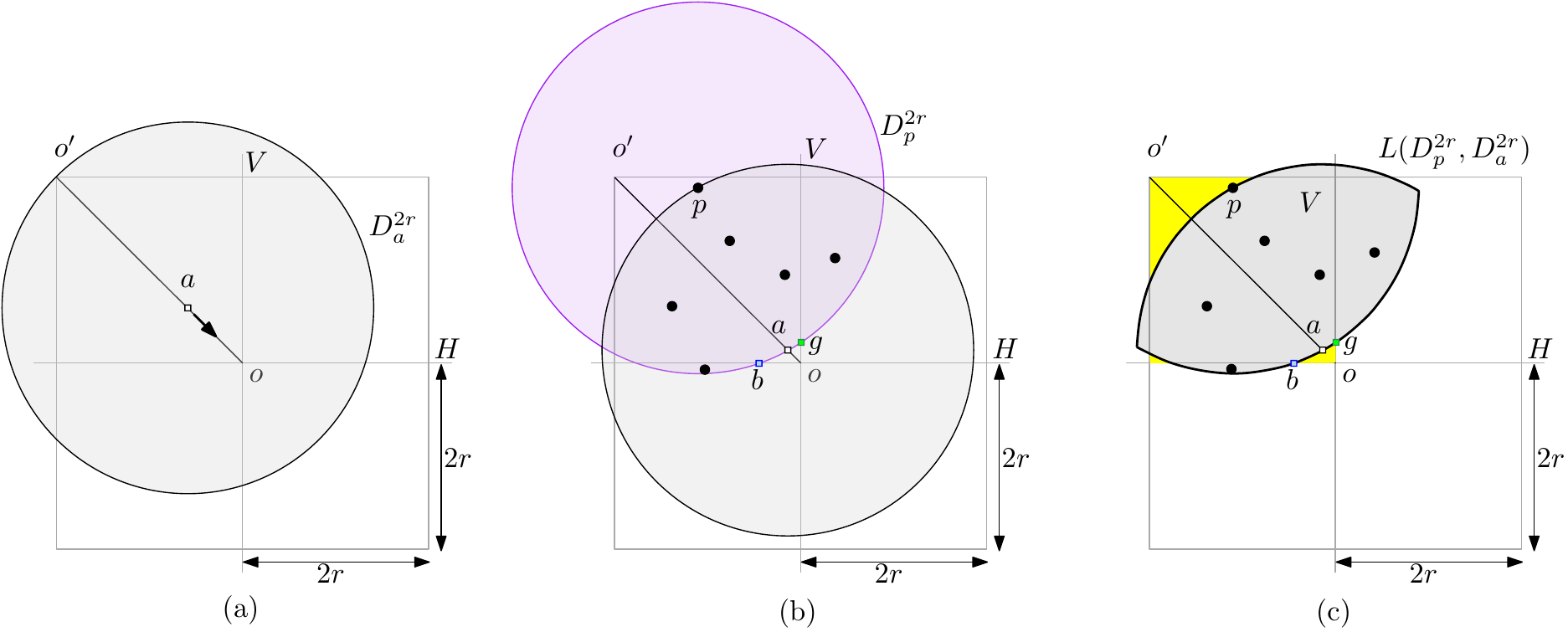}
    \caption{(a)--(b) Illustration for sweeping $D^{2r}_a$. (c) Illustration for Case A {\color{black} of Lemma~\ref{lem:subroutine}}.}
    \label{sweep}
 \end{figure}
 
\subsubsection*{Case 2 (There exists a disk center in $C$ with distance more than $2r$ to $o$) } 

Without loss of generality assume that $R_2$ contains a disk center that belongs to $C$ and has a distance of more than $2r$ from $o$. We now consider    the points of $R_2$ that have a distance of more than $2r$ from $o$ and order them by sweeping a disk as described below. Let $o'$ be the top left corner of $R_2$. Let $D^{2r}_a$ be a disk of radius $2r$ such that its boundary passes through $o'$ and its center $a$ lies on the line $oo'$ (Figure~\ref{sweep}(a)). We now move $D^{2r}_a$ along the line $o'o$ by moving the center $a$ towards $o$. We order the points  in the order $D^{2r}_a$ hits them at its boundary as we move $a$ towards $o$ (breaking ties arbitrarily). {\color{black}Figure~\ref{order}(c)} illustrates this order in blue circular arcs. Let $\sigma'$ be the resulting point  ordering.

Let $p$ be the first point of $C$ in $R_2$ that is hit by $D^{2r}_a$ at its boundary. Then the boundary of   $D^{2r}_p$ passes through $a$ (Figure~\ref{sweep}(b)). Let $b$ be the point of intersection between $D^{2r}_p$ and $H$ that is closer to $o$. Let $g$ be the point of intersection between $D^{2r}_p$ and $V$ that is closer to $o$.

In the following lemma (Lemma~\ref{lem:subroutine}), we  show that 
every point of $C$ belongs to the  lens $L({D^{2r}_p,D^{2r}_a})$. Since the corresponding intersection graph is a cobipartite graph~\cite{DBLP:journals/dm/ClarkCJ90},  a maximum clique in this graph can be computed in $O(n^{1.5}\log n)$ time~\cite{DBLP:journals/jal/AggarwalIKS91}. Over all choices for $p$ in $\sigma'$, the running time becomes $O(n^{2.5}\log n)$.

\begin{lemma}\label{lem:subroutine}
{\color{black} Let $p$ be the first point of $C$ in $R_2$ that is hit by $D^{2r}_a$ at its boundary. Then every} 
point of $C$ belongs to the  lens $L({D^{2r}_p,D^{2r}_a})$.
\end{lemma}
\begin{proof} Suppose for a contradiction that there exists a point $e\in C$ that does not belong to $L({D^{2r}_p,D^{2r}_a})$.  We now consider the following three subcases and in each case we show that such a point $e$ cannot exist.

\noindent\textbf{Case A ($e \in R_2$):}  Figure~\ref{sweep}(c) highlights the potential locations for $e$ in yellow. If $e\in (D^{2r}_a\setminus D^{2r}_p)$, then $e$ fails to intersect $p$. If $e\in (D^{2r}_p\setminus D^{2r}_a)$, then $e\in C$ must be the first point (instead of $p$) in $R_2$ that is hit by $D^{2r}_a$, which leads to a contradiction. 

\noindent\textbf{Case B ($e \in R_4$):} Since $D^r_p$ and $D^r_e$ intersect, $e$ must lie inside $ D^{2r}_p$. Note that $|po|>2r$. 
Since   $p$ belongs to $C$, we have $C\subseteq D^{2r}_p$, and hence, $R_4$ cannot contain any point of $C$.
 \begin{figure} 
     \centering
     \includegraphics[width=\textwidth]{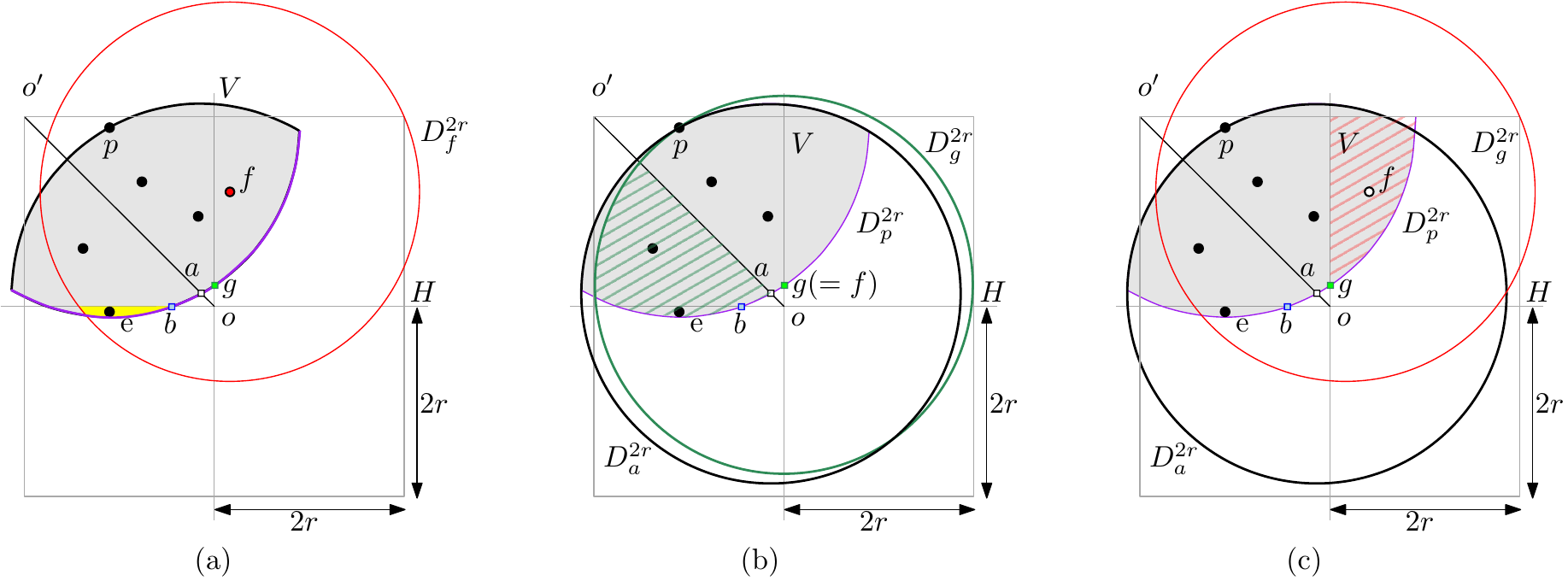}
    \caption{ Illustration for Case C  {\color{black} of Lemma~\ref{lem:subroutine}}. (a) The boundaries of $D^{2r}_f$ and $D^{2r}_p$ are  shown in red and purple, respectively. (b) The scenario when $f$ coincides with $g$. (c) The scenario when $f\not = g$.}
    \label{sweep2}
 \end{figure}
 
\noindent\textbf{Case C ($e \in R_1$ or $e \in R_3$):} Without loss of generality assume that $e \in R_3$. The argument when $e \in R_1$ is symmetric. 
We have explained in Case B that $R_4$ cannot contain any point of $C$. Therefore, by Remark~\ref{rem:opp}, $R_1$ and $R_3$ each contains a point from $C$. Let $f\in C$ be a point in $R_1$. Then $D^r_e$ must intersect $D^r_f$. In other words, $e$ must belong to $D^{2r}_f$ (Figure~\ref{sweep2}(a)). It now suffices to show that the region   $R_3 \cap L(D^{2r}_f \cap  D^{2r}_p)$, which is the potential location  for $e$ (shown in yellow),  is a subset of $L({D^{2r}_p,D^{2r}_a})$.

Consider first the case when $f$ coincides with $g$. Since $g$ is on the boundary of $D^{2r}_p$, the 
 boundary of $D^{2r}_g$ passes through $p$. Note that if we walk along the boundary of $D^{2r}_p$ {\color{black} starting at $g$} clockwise, then we first hit $g$ and then $a$. If we keep walking then we must hit the boundary of $D^{2r}_g$ before the boundary of  $D^{2r}_a$. Therefore, the region of  $L({D^{2r}_p,D^{2r}_g})$  on the left halfplane of line $oo'$ (as shown in rising pattern in Figure~\ref{sweep2}(b)) is a subset of $L({D^{2r}_p,D^{2r}_a})$. Consequently, $R_3\cap L({D^{2r}_p,D^{2r}_g})$  is a subset of $L({D^{2r}_p,D^{2r}_a})$.

Consider now the case when $f\not = g$ and $f\in (R_1\cap D^{2r}_p)$ (Figure~\ref{sweep2}(c)). By Remark~\ref{rem:triv}, we obtain $(R_3\cap D^{2r}_f) \subseteq (R_3\cap D^{2r}_g$). Together with the argument that  $R_3\cap L({D^{2r}_p,D^{2r}_g})$  is a subset of $L({D^{2r}_p,D^{2r}_a})$, one can observe that $R_3\cap L({D^{2r}_p,D^{2r}_f})$ is a subset of $L({D^{2r}_p,D^{2r}_a})$. 
\end{proof}

Since a maximum clique in $S$ can be computed in $O(n^{2.5}\log n)$ time, the strategy of Section~\ref{sec:idea} yields a running time of $O(n^{2.5} \log n)$.


\begin{theorem}
    Given a set of $n$ unit disks in the Euclidean plane, a maximum clique in the corresponding disk graph can be computed in $O(n^{2.5}\log n)$ time.
\end{theorem}

\section{Combination of Unit Disks and Axis-Parallel Rectangles}\label{UDR}


In this section we show that  the maximum clique problem for an intersection graph of unit disks and axis-parallel rectangles is NP-hard to approximate within a factor of  $\frac{4448}{4449}\approx 0.9997$. 
We first show an inapproximability result  for computing a maximum independent set and then use this result to prove the APX-hardness for computing a maximum clique.

\subsection{Inapproximability of Computing a Maximum Independent Set}

The proof of the following theorem is obtained by leveraging an inapproximability result of~\cite{DBLP:journals/tcs/ChlebikC06} and a graph transformation technique of~\cite{DBLP:journals/dm/FleischnerSS10}. 

\begin{theorem}\label{inap}
The problem of computing a maximum independent set in a 2-subdivision of a Hamiltonian cubic graph is NP-hard to approximate within  $\frac{4448}{4449}\approx 0.9997$, even when a Hamiltonian cycle is given as an input.  
\end{theorem}  
\begin{proof}
We use an NP-hard gap result of Chleb{\'{\i}}k and  Chleb{\'{\i}}kov{\'{a}}~\cite{DBLP:journals/tcs/ChlebikC06}. They showed  that for every fixed $\varepsilon >0$, it is NP-hard to decide whether a cubic graph $G$ with $n$ vertices has an independent set of size at least $\frac{n}{2}(1-2\delta+\varepsilon)$ or at most $\frac{n}{2}(1-3\delta-\varepsilon)$, where $\delta =  0.0103305$. 

\begin{figure}[h]
  \centering
  \includegraphics[width=0.75\textwidth]{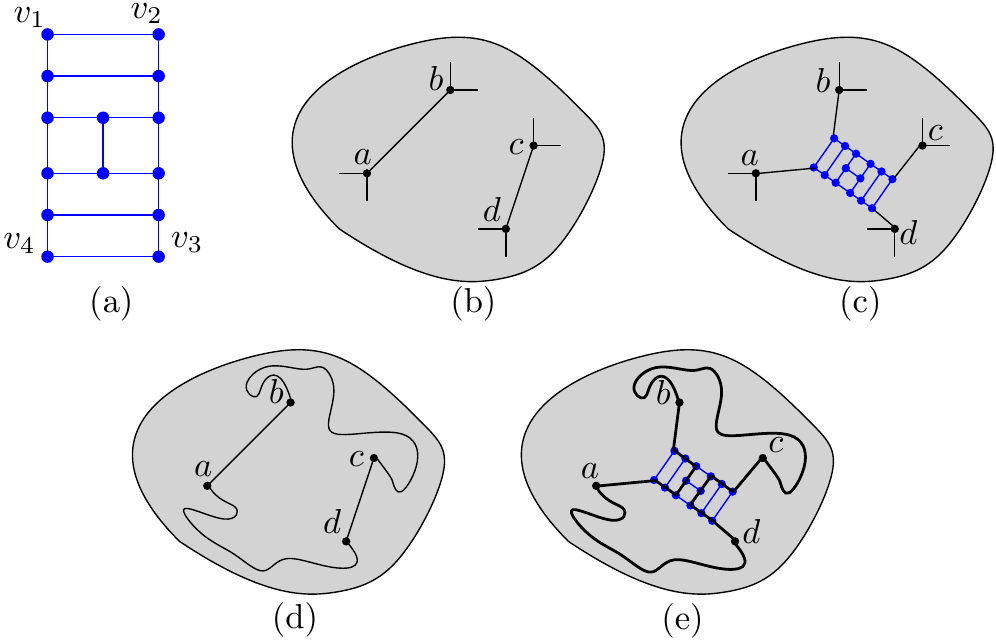}
\caption{(a) A ladder. (b)--(c) Insertion of a ladder. (d)--(e) Insertion of a ladder on a pair of Hamiltonian edges and a Hamiltonian cycle after the insertion of the ladder. 
}\label{fig:ladder}
\end{figure}

We now transform $G$ into a Hamiltonian cubic graph as follows. A \emph{ladder} is a graph obtained by taking a  $6\times 2$ grid graph, then  subdividing the edges corresponding to the middle row,  and finally, joining the division vertices with an edge, as shown in Figure~\ref{fig:ladder}(a). A \emph{ladder insertion} on a pair of edges $(a,b)$ and $(c,d)$ of $G$ is 
 an operation that removes  these edges and makes the end vertices adjacent to the 4 corners of the ladders as illustrated in  Figure~\ref{fig:ladder}(b)--(c). Fleischner et al.~\cite{DBLP:journals/dm/FleischnerSS10} showed that  $G$ can be transformed into a Hamiltonian cubic graph $H$ by a sequence of  at most $n$ ladder insertion operations and the construction yields a Hamiltonian cycle in polynomial time. We add additional  ladders to $H$ such that it contains exactly $n$ ladders. These additional leaders are inserted on pairs of Hamiltonian edges. This allows us to maintain a Hamiltonian cycle over the ladder insertions, as shown in Figure~\ref{fig:ladder}(d)--(e). Fleischner et al.~\cite{DBLP:journals/dm/FleischnerSS10} showed that each ladder insertion increases the size of the maximum independent set in the resulting graph by exactly 6.  Consequently, $\alpha(H) = \alpha(G)+ 6n$. Let $H'$ be a $2$-subdivision of $H$. Then it is known that $\alpha(H') = \alpha(H)+ m$~\cite{ChlebikC07}, where $m=\frac{3(n+14n)}{2}$ is the number of edges in $H$.  Therefore, $\alpha(H') = \alpha(G)+ \frac{45n}{2}$.

By the NP-hard gap results of~\cite{DBLP:journals/tcs/ChlebikC06}, 
for every fixed $\varepsilon >0$, it is NP-hard to decide whether   $H'$ has an independent set of size at least $\frac{n}{2}(46-2\delta+\varepsilon)$ or at most $\frac{n}{2}(46-3\delta-\varepsilon)$. Consequently, we obtain an inapproximability factor of $1- \frac{\delta}{46-3\delta} \approx 1- \frac{1}{4449.83}$. 
\end{proof}

Let $G$ be a Hamiltonian cubic graph with $n$ vertices and let $H$ be a 2-subdivision of $G$. In Section~\ref{intersecrep} we show that given a Hamiltonian cycle of $G$,  $\overline{H}$  can be represented as an intersection graph of unit disks and axis-parallel rectangles in polynomial time. Since a maximum independent set in $H$ corresponds to a maximum clique in $\overline{H}$ and vice versa, the inapproximability result follows from Theorem~\ref{inap}.

\subsection{Representing $\overline{H}$ with Unit Disks and Axis-parallel Rectangles}\label{intersecrep}

The number of edges in $G$ is $m= 3n/2$. We first show that the edges of $G$ can be oriented and labeled with distinct  positive integers from $1$ to $3n/2$ such that each vertex has exactly two of its incident edges with the same orientation and they are labeled with consecutive numbers. We will refer to such labeling as a \emph{pair-oriented labeling}. Figure~\ref{fig:pair}(i) illustrates a pair-oriented labeling of a cubic graph, e.g., $v_5$ has two incoming edges which are labeled with 2 and 3, and $v_7$ has two outgoing edges which are labeled with 8 and 9.   We will use this labeling to construct the required intersection representation for $\overline{H}$.

\begin{lemma}\label{polabel}
Let $G$ be a Hamiltonian cubic graph with $n$ vertices. Then   $G$ admits a pair-oriented labeling. Furthermore, given a Hamiltonian cycle $C$ in $G$, a pair-oriented labeling for $G$ can be computed in polynomial time.
\end{lemma}
\begin{proof}
We first orient the edges of $G$, as follows. Let $(v_1,v_2,\ldots,v_n)$ be the ordering of the vertices of $G$ on $C$.  For each edge $(v_i,v_j)$, where $i<j$, we orient the edge from $v_i$ to $v_j$, as illustrated in Figure~\ref{fig:pair}(a). We will refer to the edges on $C$ as the \emph{Hamiltonian edges} and the rest of the edges as \emph{non-Hamiltonian edges}. For a non-Hamiltonian edge $(v_i,v_j)$, we will refer to the Hamiltonian edges $(v_i,v_{i+1})$ and $(v_{j-1},v_j)$ as the \emph{nested edges} of $(v_i,v_j)$.

We now give an incremental construction for the edge labeling. We first find the smallest index $k$ such that $v_k$ has a pair of outgoing edges that are not yet labeled. We now find a \emph{maximal edge sequence} $S_k$ of non-Hamiltonian edges $e_1,e_2,\ldots,e_q$ such that for each $i$ from 1 to $q-1$, there is a Hamiltonian edge that connects the source vertex of $e_{i+1}$ to the target vertex of $e_{i}$. illustrates such a maximal edge sequence {\color{black}$e_1,e_2$, where $e_1=(v_1,v_5)$ and $e_2=(v_4,v_6)$,} and the edge $(v_4,v_5)$ is a Hamiltonian edge.   Let $\ell$ be the largest number that has been used for edge labeling so far. We then label the edges $e_1,e_2,\ldots,e_q$ with $\ell+2,\ell+4,\ldots, \ell+2q$ and the Hamiltonian edges that they nest with $\ell+1,\ell+3,\ldots, \ell+(2q+1)$. Let $V_k$ be the set of vertices that appear  on the edges of $S_k$. It is now straightforward to verify that every vertex of  $V_k$ has two edges with the same orientation and these edges are labeled with consecutive numbers.  We repeatedly find such  maximal edge sequences starting at the vertex with the smallest index that has two outgoing edges that are not yet labeled. Figure~\ref{fig:pair}(a)--(g) illustrate the labeling process where the first maximal edge sequence (that starts at $v_1$) is shown in blue, the second maximal edge sequence that starts at $v_2$ is shown in red, and so on.  

In the following we show that the above process ensures for every vertex, a pair of edges with the same orientation that are labeled with consecutive numbers.  Therefore, the remaining unlabeled edges can be labeled using the numbers available in arbitrary order (Figure~\ref{fig:pair}(h)). 

Note that for two maximal edge sequences $S_p$ and $S_q$, the non-Hamiltonian edges are disjoint by construction. Since $G$ is a cubic graph, the Hamiltonian edges that are nested by the edges of $S_p$ are also disjoint from the  Hamiltonian edges that are nested by the edges of $S_q$. Therefore, the labels used for $S_p$ and its nested edges do    not have any conflict with those that are used for $S_q$ and its nested edges. It now suffices to show that for each vertex $v$ in $G$, there is a maximal edge sequence  that contains $v$. 

\begin{figure}[h]
  \centering
  \includegraphics[width=\textwidth]{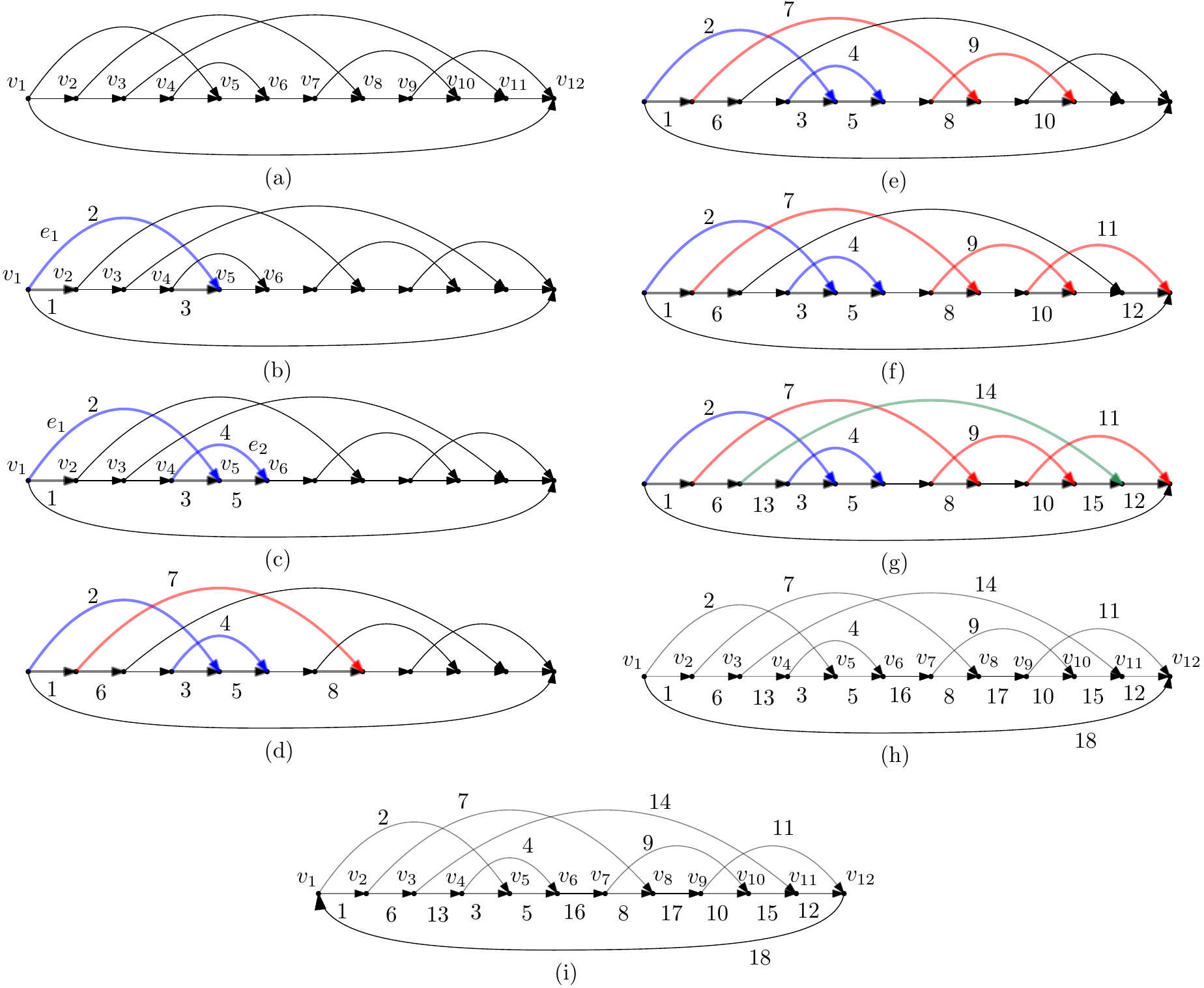}
\caption{Construction of a pair-oriented  labeling of a Hamiltonian cubic graph.  The first maximal edge sequence $S_1$ is shown in blue. $S_1$ consists of the non-Hamiltonian edges $e_1,e_2$, and the set $V_1$ consists of their end vertices, i.e., $\{v_1,v_4,v_5,v_6\}$. The second and third maximal edge sequences are shown in red and green, respectively. 
}\label{fig:pair}
\end{figure}
Suppose for a contradiction that  $j$ is the smallest index such that $v_j$ does not appear in any maximal edge sequence. 
 Then $v_j$ must have two edges $e$ and $e'$ that are not labeled. We now consider the following cases.
 
\noindent\textbf{Case 1 ($e,e'$ are outgoing edges of $v_j$):} If $e,e'$ are outgoing edges of  $v_j$, then it must be found as we repeatedly search for the vertex with the smallest index that has two unlabeled outgoing edges. 

\noindent\textbf{Case 2 (One of $e,e'$ is an incoming edge of  $v_j$ and the other is an outgoing edge of $v_j$):} Assume without loss of generality that $e$ is incoming and $e'$ is outgoing. Let $e''$  be the third edge of $v_j$. 
\begin{enumerate}[]
    \item \textbf{Case 2.1:} If $e''$  
  is an outgoing non-Hamiltonian edge,  then it must be labeled; otherwise, we get two unlabeled outgoing edges and can apply Case 1. Since $e''$ nests one of $e,e'$, and one of them must be labeled, which leads to a contradiction. 
    \item \textbf{Case 2.2:} If $e''$  
  is an incoming non-Hamiltonian edge $e=(v_i,v_j)$, then we have $i<j$.   
  If $e''$ is labeled, then since $e''$  nests one of $e,e'$, we must already have one of $e,e'$   labeled. Consider now the scenario when  $e''$ is not labeled.
  Since the maximal edge sequences are disjoint, the nested edge $(v_i,v_{i+1})$ cannot be a part of any maximal edge sequence. 
 Hence  $(v_i,v_{i+1})$ is unlabeled. Since   $v_i$ has two unlabeled outgoing edges and since $i<j$,  it   contradicts our initial choice of $v_j$. 
  \item \textbf{Case 2.3:} If $e''$  
  is an incoming Hamiltonian edge, then $e$ is an incoming non-Hamiltonian edge. Without loss of generality assume that 
   $e=(v_i,v_j)$ and the rest of the argument is similar to Case 2.2. 
 \item \textbf{Case 2.4:} If $e''$  
  is an outgoing Hamiltonian edge, then $e'$ is an outgoing  non-Hamiltonian edge. Here $e''$ must be labeled; otherwise, we get two unlabeled outgoing edges and can apply Case 1. Therefore, the maximal edge sequence that puts a label on $e''$ must contain $e'$, and thus $e'$ cannot be unlabeled.

\end{enumerate}

\noindent\textbf{Case 3 ($e,e'$ are incoming edges of $v_j$):} 
 If $e,e'$ are incoming edges of  $v_j$, then we consider two cases depending on whether both are Hamiltonian or not. 

\begin{enumerate}[]
\item  \textbf{Case 3.1:} If one of $e$ and $e'$ is non-Hamiltonian, then  assume without loss of generality that $e=(v_i,v_j)$ is non-Hamiltonian. We now can reach a contradiction in the same way as we argued in Case 2.2. 

\item \textbf{Case 3.2:}   If both  $e,e'$ are Hamiltonian, then $v_j$ coincides with $v_n$.  Thus the remaining incident edge $e''$ of $v_j$ is an incoming non-Hamiltonian  edge. 
Here we can reach a contradiction in the same way as we argued in Case 2.2. 
\end{enumerate}

{\color{black} Note that the above case analysis ensures    for every vertex, a pair of edges with the same orientation that are labeled with consecutive numbers. By the initial orientation of the edges, every vertex except $v_1$ and $v_n$ has exactly two edges with the same orientation. To obtain the required pair-oriented labeling, it now suffices to reverse the direction of $(v_1,v_n)$, as shown in Figure~\ref{fig:pair}(i). }
\end{proof}

\begin{figure}[pt]
  \centering
  \includegraphics[width=\textwidth]{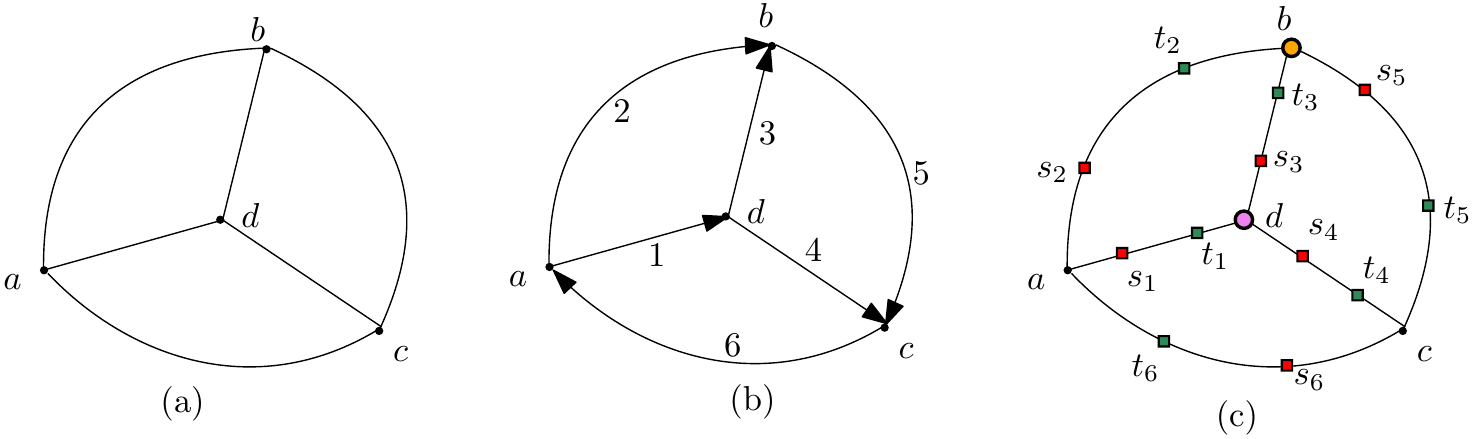}
\caption{(a)--(b)  $G$, and its pair-oriented labeling. (c) Labeling of the division vertices of $H$.
}\label{fig:setup}
\end{figure}

Let $\gamma$ be a pair-oriented labeling of $G$ (Figure~\ref{fig:setup}(a)--(b)). By $\gamma(e)$ we denote  the label of an edge $e$. Let $e$ be an edge of $G$ with source $s$ and target $t$. We now label the two division vertices corresponding to $e$ in the 2-subdivision $H$. The division vertex adjacent to $s$ receives the label  $s_{\gamma(e)}$ and the division vertex adjacent to $t$ receives the label  
 $t_{\gamma(e)}$. We refer to $s_{\gamma(e)}$ and $t_{\gamma(e)}$ as a \emph{type-$s$} and \emph{type-$t$} label, respectively. By the property of $\gamma$ (Lemma~\ref{polabel}), each original vertex $v$ in $H$ is now adjacent to exactly two division vertices of the same type with their indices numbered with consecutive numbers. For example in Figure~\ref{fig:setup}(c), the vertices  $d$ and $b$ are adjacent to $s_3,s_4$ and $t_2,t_3$, respectively.


\begin{figure}[pt]
  \centering
  \includegraphics[trim=0 170 40 0,clip, width= \textwidth]{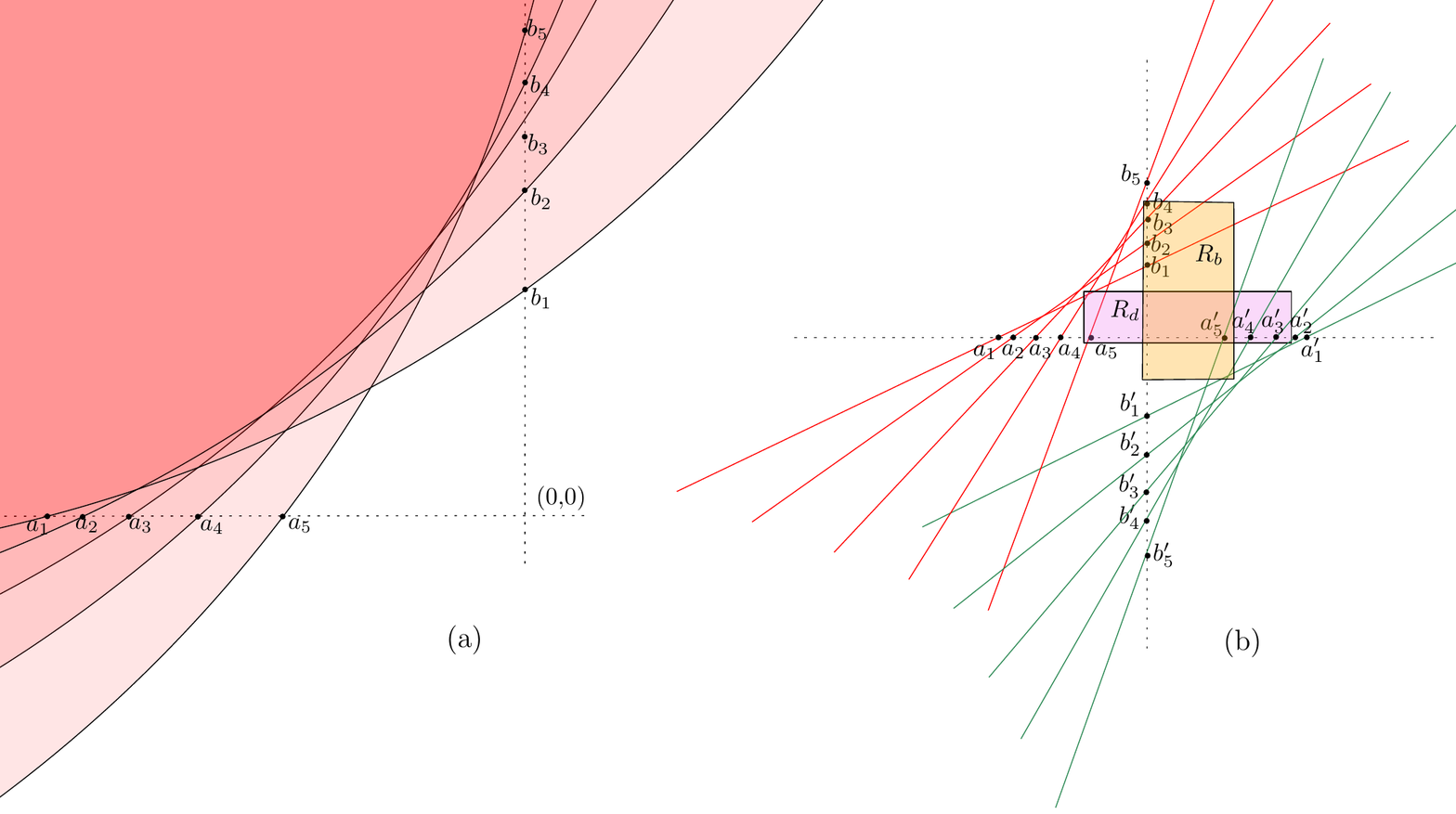}
\caption{(a) Arrangement of the unit disks corresponding to type-$s$ division vertices. (b) Illustration for the intersection representation of $\overline{H}$. Only a subset of unit disks and the rectangles corresponding to $b$ and $d$ are shown for better readability.
}\label{fig:intersec}
\end{figure}

The intersection representation of $\overline{H}$ now follows from the construction of~\cite{bonnet_et_al:LIPIcs:2020:13258}. We briefly describe the construction at a high level   for completeness. Consider a set of unit disks $C_s$ for the type-$s$ vertices with centers in the second quadrant such that they intersect the negative x-axis and positive y-axis, but   not   the positive x-axis or negative y-axis. Furthermore, the ordering of the disks obtained by walking from $(0,0)$ to {\color{black}$(-\infty,0)$}  is  reversed when walking from $(0,0)$ to $(0,\infty)$ (Figure~\ref{fig:intersec}(a)). Let $\mathcal{I}_s$ be the convex region determined by the intersection of all   type-$s$ disks. The set of unit disks $C_t$ for type-$t$ vertices is placed on the 4th quadrant symmetrically. Let $\mathcal{I}_t$ be the convex region determined by the intersection of all the type-$t$ disks. For a sufficiently large radius, the disk boundaries appear similar to a set of halfplanes (Figure~\ref{fig:intersec}(b)) and each disk in $C_s$ intersects every disk in $C_t$ except for the one with the same label. Therefore, all the intersections between division vertices of $\overline{H}$ are realized. 
Note that the original vertices of $H$ {\color{black} in $\overline{H}$} form a clique in $\overline{H}$ and each original vertex is adjacent to all but three division vertices in $\overline{H}$. We now represent the original vertices of $\overline{H}$ with rectangles such that all of them enclose the point $(0,0)$. Let $b$ be an original vertex of $\overline{H}$. Without loss of generality assume that  $b$ has two type-$t$ neighbors and one type $s$ neighbor (Figure~\ref{fig:setup}(c)). By the property of the pair-oriented labeling, the type-$t$ neighbors are labeled consecutively.  Let $t_{i},t_{i+1},s_j$ be the neighbors of $b$. We now create a rectangle $R_b$ to represent $b$. We place the top-left corner of  $R_b$ near the intersection point of the boundaries of the  disks for $t_{i},t_{i+1}$ such that $R_b$ does not intersect these disks but intersects all other type-$t$ disks. We place the bottom-right corner of  $R_b$ near the circular segment determined by the disk for $s_{j}$ on $\mathcal{I}_t$  such that $R_b$ does not intersect the disk for $s_{j}$ but intersects all other type-$s$ disks. We refer  to~\cite{bonnet_et_al:LIPIcs:2020:13258} for a formal reduction. The reason that Bonnet et al.~\cite{bonnet_et_al:LIPIcs:2020:13258} thought such a reduction with a co-2-subdivision approach is unlikely is that without the pair-ordered labeling, the rectangle for $R_b$ may require to avoid three arbitrary disks. However,   such realization may not be possible for all the rectangles if the disks are not carefully ordered. Our main contribution here is to show that a nice ordering for the disks exists to make the  co-2-subdivision approach to work, which at the same time improves the inapproximability factor of~\cite{bonnet_et_al:LIPIcs:2020:13258}.

\begin{theorem}
The problem of computing a maximum clique in an intersection graph of unit disks and axis-parallel rectangles is NP-hard to approximate within a factor of $\frac{4448}{4449} \approx 0.9997$.
\end{theorem}

\section{Finding a Maximum Clique in an  $(\varepsilon,\beta)$-disk graph}
\label{sec:epsbeta}
 
In this section, we give a polynomial-time algorithm to compute a maximum   clique in an $(\varepsilon,\beta)$-disk graph. By definition, the radii of the disks are in $[1,1+\varepsilon]$ and every lens is of width at least $\beta$. We give an $O(n^4)$-time algorithm when  $\beta \ge 0.265$ and $\varepsilon \le 0.0001$. Although $\beta$ could be expressed as a function of $\varepsilon$, for simplicity of the presentation, we set specific values to $\beta$ and $\varepsilon$ and often use crude bounds to simplify the arguments. Therefore, we believe one can choose slightly better parameters by using a tedious case analysis. 

Let $G$ be an $(\varepsilon,\beta)$-disk graph. 
Let $M$ be a set of disks determining a maximum clique in $G$. In the following, we will use Helly's theorem~\cite{Helly}, i.e.,   for  a collection of convex sets in $\mathbb{R}^d$, if the intersection of every $(d+1)$ of these sets is non-empty then the   collection must have a non-empty intersection. Recall that for three disks $\{D_a,D_b,D_c\}$, if $(D_a\cap D_b\cap D_c)   = \emptyset$, then we call them a \emph{non-Helly  triple}. Otherwise, we refer to them as a \emph{Helly triple}.  

Consider  three unit disks with centers at the corners of an equilateral triangle. If these disks intersect exactly at one point, then the width of each lens is $(2-2\cos 30^\circ) \approx 0.2679 $, which is larger than $\beta$. Therefore, in an $(\varepsilon,\beta)$-disk graph, we may have non-Helly triples. We now consider two cases depending on whether $M$ contains a non-Helly triple or not. 

\subsection{$M$ does not contain any non-Helly triple
}\label{sec:nonht}

If $M$ does not contain any non-Helly triple, then by Helly's theorem~\cite{Helly},  the disks in $M$ have a non-empty intersection. 
Let $\mathcal{R}$ be the set of connected regions or cells determined by the arrangement of the disk boundaries. We  examine for each cell $r$ in $\mathcal{R}$, the number of  disks that contains $r$, and find a maximum set of mutually intersecting disks.  It is straightforward to compute the arrangement of  $n$ disks in $O(n^3)$-time (even faster algorithms exist~\cite{DBLP:books/el/00/AgarwalS00a}) and each cell can be checked in $O(n)$ time. Hence finding a maximum  clique takes $O(n^4)$ time.

\begin{figure}[pt]
  \centering
  \includegraphics[width= \textwidth]{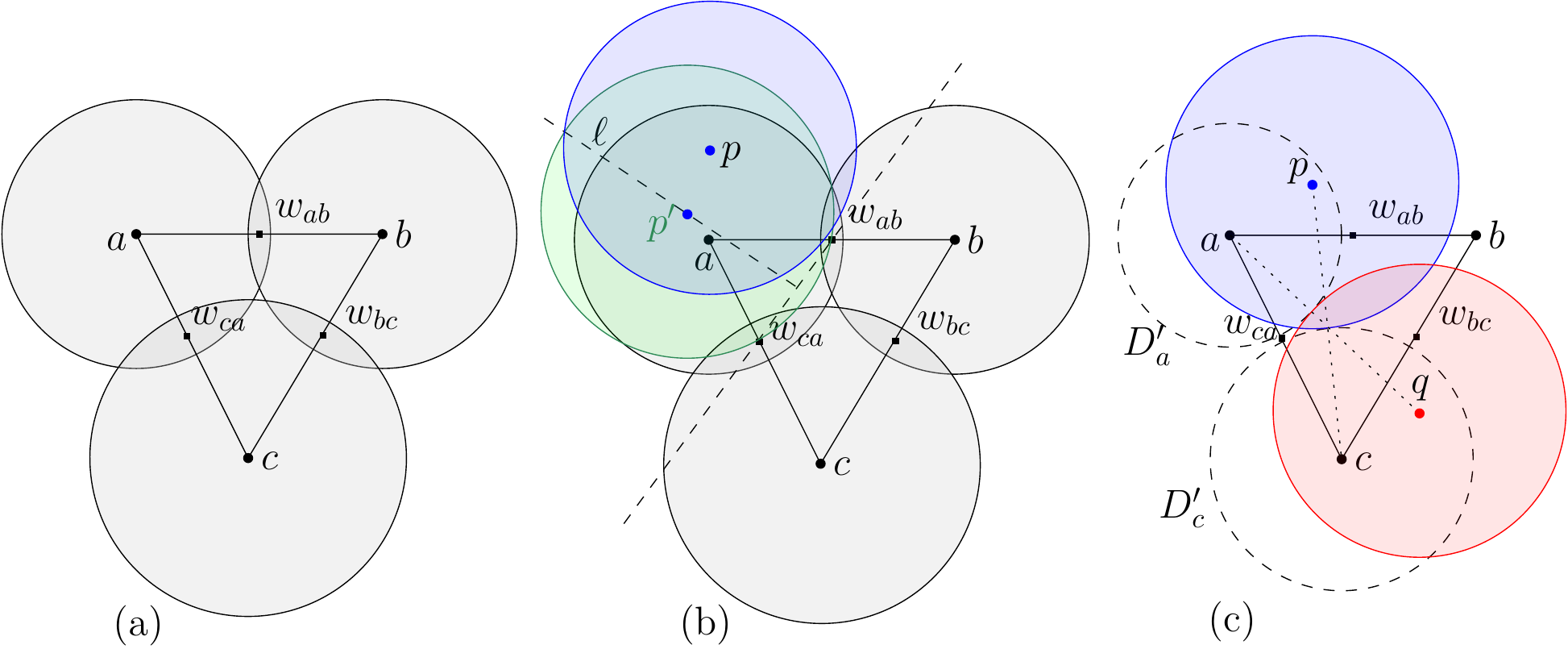}
\caption{(a) A non-Helly triple $\{D_a,D_b,D_c\}$.  Illustration for (b) 
 Case 1 and (c) Case 2.
}\label{fig:type}
\end{figure}

\subsection{$M$ contains a non-Helly triple}
If $M$ contains a non-Helly triple, then let $\{D_a, D_b, D_c\}$ be such  a non-Helly triple in $M$ (Figure~\ref{fig:type}(a)). Let $w_{ab}$,  $w_{bc}$, $w_{ca}$ be the midpoint of the lenses $L(D_a,D_b),L(D_b,D_c)$ and $L(D_c,D_a)$, respectively. The following two lemmas give some properties corresponding to the non-Helly triple. 

\begin{lemma}\label{shrink}
If $\{D_a, D_b, D_c\}$ is a non-Helly triple, $\beta \ge 0.265$, and $\varepsilon \le 0.0001$, then the lenses $L(D_a,D_b),$ $L(D_b,D_c)$ and $L(D_c,D_a)$ are of width less than  0.275. Furthermore, the interior angles of $\Delta abc$ are in the interval $[58.024^\circ,60.988^\circ]$.
\end{lemma} 
\begin{proof} 
Without loss of generality assume that the width of $L(D_a,D_b)$ is at least  0.275. Then $|ab|\le (2+2\varepsilon-0.275)=(1.725+2\varepsilon)$, and the distances $|bc|$ and $|ca|$ are upper bounded by $(2+2\varepsilon-0.265)=(1.735+2\varepsilon)$. We can restrict our attention to the triangle $\Delta abc$ when the equalities are realized. Any other configuration can be realized by moving the disks closer to each other and hence would increase the width of the lenses. 

We now show that the circumradius for $\Delta abc$ is at most  $1$, which  would contradict that $\{D_a, D_b, D_c\}$ is a non-Helly triple and hence the width of the lenses must be less than 0.275.

The circumradius of the circumcircle of $\Delta abc$ is $\frac{|ab|\cdot |bc|\cdot |ca|}{4\cdot Area(\Delta abc)}$, which is at most 
\begin{align*}
&\frac{(1.725+2\varepsilon)(1.735+2\varepsilon)^2}{  \sqrt{(|ab|+|bc|+|ca|)(-|ab|+|bc|+|ca|)(|ab|-|bc|+|ca|)(|ab|+|bc|-|ca|)}} \\ 
&\le \frac{(1.725+2\varepsilon)(1.735+2\varepsilon)^2}{  \sqrt{(5.195+6\varepsilon)(1.745+2\varepsilon)(1.725+2\varepsilon)(1.725+2\varepsilon)}}\le 1, \text{when $\varepsilon \le 0.0001$.}  \\
\end{align*} 

We now consider $\angle abc$. Here, $\cos \angle abc = \frac{|ab|^2+|bc|^2-|ca|^2}{2|ab||bc|}$.  Since  $1.725 \le |ab|,|bc|,|ca| \le (1.735+2\varepsilon)$, we now have  $\cos \angle abc \ge \frac{2(1.725)^2 -(1.735+2\varepsilon)^2}{2(1.735+2\varepsilon)^2}\ge 0.485$. 
We thus have $\angle abc \le 60.988^\circ$. Consequently, we can obtain a $180^\circ-2\cdot 60.988^\circ\ge 58.024^\circ$ lower bound on $\angle abc$.
\end{proof}

\begin{lemma}\label{prop}
If $\{D_a, D_b, D_c\}$ is a non-Helly triple, $\beta \ge 0.265$, and $\varepsilon \le 0.0001$, then 
 $\Delta abc$ and $\Delta w_{ab}w_{bc}w_{ca}$ satisfy the following properties. (a) $1.725 \le |ab|,|bc|,|ca| \le (1.735+2\varepsilon)$. (b) For each point $q$ in $\{w_{ab},w_{bc},w_{ca}\}$, the distances from $q$ to the center of the two disks containing $q$ are in  the interval $[0.8625, (0.8675+\varepsilon)]$. (c) The length of each side of $\Delta w_{ab}w_{bc}w_{ca}$ is in the interval $[0.883, 0.887]$.
\end{lemma}

\begin{proof}
(a) Consider the disks $D_a$ and $D_b$. The  sum of their radii  lies in the interval $[2,2+2\varepsilon]$. Since the width of their lens is at most $0.275$ (Lemma~\ref{shrink}) and at least $\beta$,  we have $2-0.275 = 1.725\le |ab|\le (2+2\varepsilon - \beta) \le (1.735+2\varepsilon)$.  

(b) Consider the distances $|aw_{ab}|$ and $|bw_{ab}|$. By Lemma~\ref{shrink},  $|aw_{ab}|$ and $|bw_{ab}|$ is at least $1-(0.275/2) = 0.8625$. Therefore, $0.8625\le |aw_{ab}|,|bw_{ab}| \le (1+\varepsilon-\beta/2)=(0.8675+\varepsilon)$. 

(c) Without loss of generality assume that the maximum side length of the triangle $\Delta w_{ab}w_{bc}w_{ca}$ is 
$ |w_{ab}w_{bc}|$. Then  $ |w_{ab}w_{bc}| = \sqrt{|bw_{ab}|^2 + |bw_{bc}|^2  - 2|bw_{ab}| |bw_{bc}| \cos \angle abc}$ $ \le $ $ \sqrt{2( 0.8675+\varepsilon)^2  - 2(0.8625^2)\cos 61^\circ } $ $ \le 0.887$. Similarly, the minimum side length is at least $\sqrt{|bw_{ab}|^2 + |bw_{bc}|^2  - 2|bw_{ab}| |bw_{bc}| \cos \angle abc}$ $ \ge $ $ \sqrt{2( 0.8625)^2  - 2(0.8675+\varepsilon)^2\cos 58^\circ } $ $ \ge 0.883$. 
\end{proof}

Let $O$ be the disks in the {\color{black} disk graph representation} and let $O'$ be $O\setminus \{D_a,D_b,D_c\}$.  We refer to a disk in $O'$  as \emph{type-$k$}, where $0\le k\le 3$, if it contains exactly $k$ points from $\{w_{ab},  w_{bc}, w_{ca}\}$. In the following we  show that for a pair of disks $D_p,D_q$, if each of them intersects {\color{black}all the disks in $\{D_a,D_b,D_c\}$}, then they must mutually intersect. As a consequence, we can find a maximum clique including $\{D_a,D_b,D_c\}$ in $O(n)$ time and a maximum clique over all possible $O(n^3)$ choices of non-Helly triples in $O(n^4)$ time.   

\smallskip
\noindent
\textbf{Case 1 (At least one of $D_p$ and  $D_q$ is of Type-0): } We show that this case is trivial because a type-0 disk that intersects all the disks in the non-Helly triple but avoids $\{w_{ab},  w_{bc}, w_{ca}\}$ cannot exist. 

Consider the disk $D_p$. We first show that if $p$ lies inside $\Delta w_{ab}w_{bc}w_{ca}$, then $D_p$ must contain a corner of $\Delta w_{ab}w_{bc}w_{ca}$. 
 By Lemma~\ref{prop}, the maximum side length of $\Delta w_{ab}w_{bc}w_{ca}$ is at most $0.887$. Therefore, the circumradius for $\Delta w_{ab}w_{bc}w_{ca}$ is bounded by 
$\frac{0.887}{\sqrt{3}}\le 1$. 
 Hence $D_p$ must contain a corner of $\Delta w_{ab}w_{bc}w_{ca}$.

We now show that if $p$ lies outside of $\Delta w_{ab}w_{bc}w_{ca}$, then $D_p$ cannot create a lens of width $\beta$ with $D_a, D_b, D_c$. Without loss of generality assume that the left-halfplane of the line through $w_{ab}w_{ca}$ contains $p$ and the right-halfplane contains the centers $b,c$ (Figure~\ref{fig:type}(b)). 

Consider a disk $D_{p'}$ with the same radius as that of $D_p$ such that its center $p'$ lies outside of $\Delta abc$ and its boundary passes through $w_{ab}$ and $w_{ca}$. The following lemma gives an upper bound on $|ap'|$ and $\angle acp'$. 

\begin{lemma}\label{lem:ap}
Let $D_{p'}$ be a disk such that the boundary of $D_{p'}$ passes through $w_{ab}$ and $w_{ca}$ and the center $p'$ lies outside of $\Delta abc$. Then $|ap'|< (0.25+\varepsilon)$ and $\angle acp'   \le 17.5^\circ$.    
\end{lemma}
\begin{proof}
To compute an upper bound on $|ap'|$, we will use the upper and lower bounds on $|aw_{bc}|$, $|aw_{ab}|$, and  $\angle abc$. Note that $0.8625\le |aw_{bc}| , |aw_{ab}| \le (0.8725 + \varepsilon)$, and by Lemma~\ref{shrink}, $\angle abc\in [58.024^\circ,60.988^\circ]$. Therefore, the potential location of $a$ is inside the intersection region of two annuli, where one is centered at $w_{ab}$  and the other is centered at $w_{ca}$. Figure~\ref{fig:geom} illustrates this region in shaded gray.


Let $t$ be a point such that $|tw_{ca}| = |tw_{ab}| = 0.8625$.  Let $h$ and $u(=w_{ab}w_{ca})$ be the height and base of   $\Delta tw_{ab}w_{ca}$, respectively. Let $h'$ and $u$ be the base  of $\Delta p'w_{ab}w_{ca}$. 
Since the boundary of $D_{p'}$  passes through $w_{ab}$ and $w_{ca}$, we have $|tp'| = h'-h 
\le \sqrt{(1+\varepsilon)^2 - (|w_{ab}w_{ca}|^2/4)} - \sqrt{0.8625^2 - (|w_{ab}w_{ca}|^2/4)}$, which increases with the increase in  $|w_{ab}w_{ca}|$. Therefore, $|tp'|    \le   0.1835$.

\begin{figure}[h]
  \centering
  \includegraphics[width=0.60\textwidth]{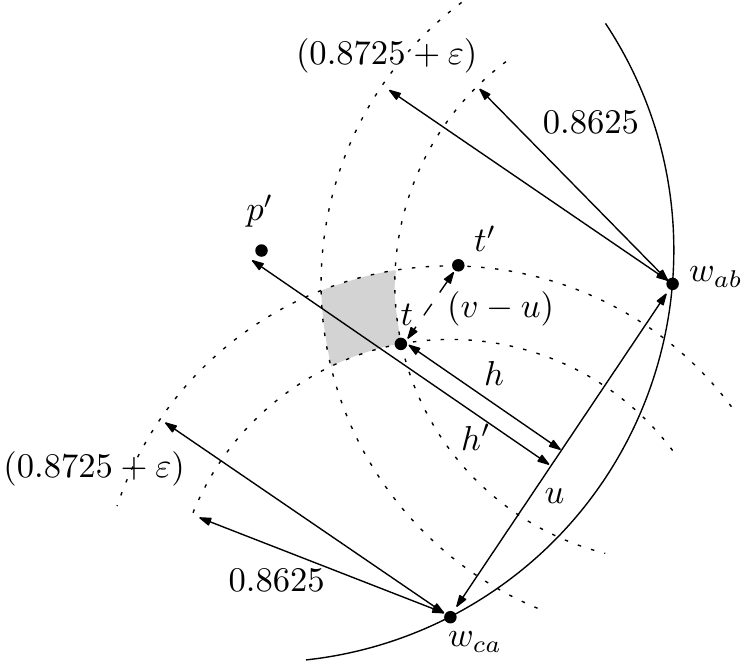}
\caption{Illustration for Lemma~\ref{lem:ap}. The potential location for $a$ is shaded in gray.
}\label{fig:geom}
\end{figure}

Consider now another triangle    
 with height $h$, base $v$, and side length $|t' w_{ca}|=(0.8725+\varepsilon)$, then  $v = \sqrt{4(|t' w_{ca}|^2-h^2)}$. The difference  in the length of the bases  is as follows:  $(v-u)=\sqrt{4((0.8725{+}\varepsilon)^2 {-} (0.8625^2 {-}  |w_{ab}w_{ca}|^2/4 ))} - |w_{ab}w_{ca}| {\le} \sqrt{4((0.8725{+}\varepsilon)^2{-} 0.8625^2 {+}  0.887^2/4)} - 0.883 < 0.0427$. 
By triangle inequality, $|ap'|\le (v-u)+|p't| \le (0.1835+\varepsilon)+0.0427 = (0.2261+\varepsilon)$.




We now compute an upper bound on $\angle acp'$.  
By Lemma~\ref{prop}, $1.725 \le |bc| \le (1.735+2\varepsilon)$. 
Since $|ap'|< (0.25+\varepsilon)$ and since $|cw_{ca}|,|bw_{ab}|\le (0.8675+\varepsilon)$, we have $1.725-(0.25+\varepsilon) \le |p'b|,|p'c| \le (1+\varepsilon)+(0.8675+\varepsilon)$. 
Therefore, $\cos\angle p'cb = \frac{|p'c|^2+|bc|^2-|p'b|^2}{2|p'c||bc|} \ge \frac{1.475^2+1.725^2 -(1.8675+2\varepsilon)^2}{2(1.8675+2\varepsilon)(1.735+2\varepsilon)}\ge 0.25$. 
We thus have $\angle p'cb \le 75.52^\circ$ and $\angle acp' \le (75.52^\circ-58.024^\circ) \le  17.5^\circ$.
\end{proof}

%
%
%
%
%

We now show that $D_{p'}$ cannot create a lens of width $\beta$ with $D_c$. Since $|p'w_{ca}|$ is fixed, the distance $|p'c|$ decreases with the increase in $\angle p'cw_{ca}$ and decrease in $|cw_{ca}|$. Since  $\angle p'cw_{ca}<  17.5^\circ$ and $|cw_{ca}|\ge 0.8625$, by using basic trigonometry on $\Delta p'cw_{ca}$ one can observe that $|p'c|\ge 1.78 > (2-\beta)$. Therefore, $D_{p'}$ cannot create a lens of width $\beta$ with $D_c$.


Since $D_p$ does not contain $w_{ab}$ and $w_{bc}$,   $p$ lies above or below the bisector $\ell$  of $w_{ab}w_{ca}$. Consider moving $p'$ to $p$. Since moving $p'$ above or below $\ell$ decreases the width of either  $L(D_{p'}, D_b)$ or $L(D_{p'}, D_c)$,  $D_p$ cannot have a lens of width $\beta$ with  $D_b$ and $D_c$ simultaneously.


\smallskip
\noindent
\textbf{Case 2 ($D_p$ and $D_q$ are of Type-1): }  Without loss of generality assume that $D_p$ and $D_q$ contains $w_{ab}$ and $w_{bc}$, respectively (Figure~\ref{fig:type}(c)). Let $a_r,c_r,p_r,q_r$ be the radii of $D_a$, $D_c$, $D_p$, $D_q$, respectively. It now suffices to show that $|pq|\le p_r+q_r$, i.e., $D_p$ and $D_q$ must intersect. Note that by the property of $(\varepsilon,\beta)$-graph, an intersection would imply a lens of width at least $\beta$, and hence we only show that $|pq|\le p_r+q_r$.

Let $D'_a$ and $D'_c$ be the disks obtained by shrinking the radii of $D_a$ and $D_c$ by $\beta$. Since the width of the lenses created by the non-Helly triple is less than 0.275, the points $w_{ab}, w_{bc}, w_{ca}$ lie outside of  $D'_a$ and $D'_c$. 
Since the width of each lens  is at least $\beta$,  $D_p$ must intersect $D'_c$. Consider a line $\ell$ through $ac$ with $b$ on its right half-plane. 

Consider first the scenario when $p$ and $q$ are   on the right half-plane of $\ell$. If $p$ is above the line through $bc$ and $q$ is below the line through $ab$, then $aq$ and $pc$ intersect (Figure~\ref{fig:type2}(a)). Therefore, $|pq|{\le}|aq|+|pc|-|ac| {\le} (a_r+q_r-0.265){+}(p_r+c_r-0.265){-}(a_r+c_r-0.275) {<}  (p_r+q_r)$.  Otherwise, $p,q$ lie on the right halfplane of the line through $w_{ab}w_{bc}$ in the wedge determined by $\angle abc$ and its opposite angle, as shaded in orange in Figure~\ref{fig:type2}(b). Since $\max\{|w_{ab}w_{bc}|,|bw_{ab}|,|bw_{bc}|\}\le 1+\varepsilon$, it is straightforward to observe that  $D_p$ and $D_q$ intersect.

\begin{figure}[h]
  \centering
  \includegraphics[width= 0.98\textwidth]{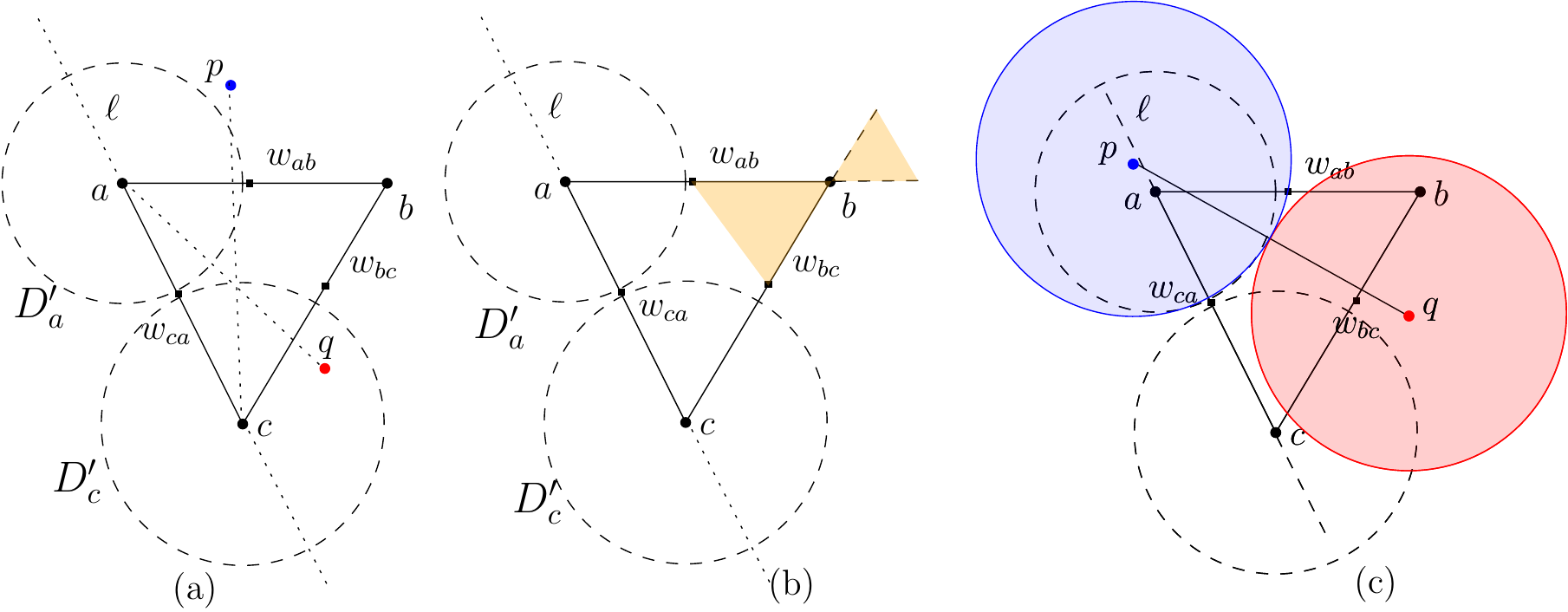}
\caption{  Illustration for the locations of $p$ and $q$. (a)--(b) Case 1. (c) Case 2.}\label{fig:type2}
\end{figure}

If $p$ and $q$ are on different sides of $\ell$ (Figure~\ref{fig:type2}(c)), then without loss of generality assume that $p$ lies on the left half-plane and $q$ lies on the right half-plane. 
In the following, we will show that $|ap|\le (0.25-\varepsilon)$.  Consequently, $|pq| \le |ap| + |aq|\le  (0.25-\varepsilon)+ (a_r+q_r-0.265) = (p_r+q_r) +(a_r-p_r)-\varepsilon  -0.015 < (p_r+q_r)$. 

We now show that $|ap|\le (0.25-\varepsilon)$. For a fixed angle $\angle w_{ab}aw_{ca}$, the distance $|ap|$ is maximized when the boundary of $D_p$ passes through $w_{ab}$  (Figure~\ref{fig:type2}(c)). Note that $D_p$ must avoid $w_{ca}$. We now have $|ap| = \sqrt{|pw_{ab}|^2{+}|aw_{ab}|^2{-}2|pw_{ab}||aw_{ab}|\cos \angle pw_{ab}a}$ $ \le  \sqrt{ (1+\epsilon)^2{+}|aw_{ab}|^2-  2|aw_{ab}|\cos \angle pw_{ab}a}$. We obtain $|ap| \le (0.25-\epsilon) $ when $ \angle pw_{ab}a \le 12.61^\circ$ and $0.85\le|aw_{ab}|\le 1.09$. Since by Lemma~\ref{prop} we already have  $0.8625\le|aw_{ab}|\le (0.8675+\epsilon)$, it suffices to show that $ \angle pw_{ab}a \le 12.61^\circ$.

For a fixed $|aw_{ab}|$, the $\angle pw_{ab}a$ is maximized when $\angle bac$ is maximum and $|pw_{ab}| = (1+\epsilon)$. By Lemma~\ref{shrink}, $\angle bac \le 60.988^\circ$. Furthermore, moving $w_{ab}$ towards $a$ increases the angle $\angle pw_{ab}a$. Therefore, the maximum of $\angle pw_{ab}a$ is attained when $|aw_{ab}|=0.8625$ and $\angle paw_{ab} = 180^\circ - 60.988^\circ = 119.012^\circ$. Since we now know the lengths of two sides of $\Delta pw_{ab}a$ and the $\angle paw_{ab}$, it is straightforward to compute $\angle pw_{ab}a$, which is smaller than $12.61^\circ$, as required. 


\smallskip
\noindent
\textbf{Case 3 ($D_p$ and $D_q$ are of Type-2 or Type-3): } Since     $D_p$ and $D_q$ each contains at least two points from $\{w_{ab},  w_{bc}, w_{ca}\}$, they must intersect. 

\smallskip
\noindent
\textbf{Case 4 (One of $D_p$ and $D_q$ is of type-1 and the other is of type-2 or type-3): } The case when   $D_p$ and $D_q$ contains a common point from  $\{w_{ab},w_{bc},w_{ca}\}$ is trivial. Therefore, without loss of generality 
assume that $D_p$ is of type-1 and contains $w_{ab}$, and $D_q$ is of type-2 and contains $w_{bc}$ and $w_{ca}$.  
We use the same setting as in Case 2, i.e., $\ell$ is the line through $ac$ and $b$ lies on the right half-plane. We now {\color{black}move} $D_q$ counter-clockwise without changing the distance of $|pq|$ and stop as soon as   $w_{ca}$ hits the boundary of $D_q$. By an analysis similar to Case 2, we now can observe that $|pq|\le p_r+q_r$, and hence $D_p$ and $D_q$ must intersect.

\begin{theorem}
Given a set of $n$ disks in the Euclidean plane such that the width of every lens is at least $0.265$ and the radii are in  the interval $[1,1.0001]$, a maximum clique in the corresponding disk graph can be computed in $O(n^4)$ time.
\end{theorem}
\section{Conclusion and Directions for Future Work}

We gave an $O(n^{2.5}\log n)$-time algorithm to compute a maximum clique in a unit disk graph. A natural avenue for future research would be to improve the time complexity of the algorithm. 
We   proved  that for the combination of unit disks and axis-parallel rectangles, a maximum clique is NP-hard to approximate within a factor of $4448/4449$. We obtained the result using a co-2-subdivision approach, and along the way,  we showed that every Hamiltonian cubic graph admits a pair-oriented labeling. It would be interesting to improve the inapproximability factor, and one way to achieve this would be to examine whether pair-oriented labelings exist also for non-Hamiltonian cubic graphs. 
We showed that if the width of every lens is at least  0.265, then one can find a maximum clique in polynomial time  in a more general setting where the disk radii are in $[1, 1.0001]$. We believe that  {\color{black}with tedious case analysis}, 
these numbers may be improved slightly, however, it would be challenging to lower $\beta$ down to $0.2$ using the current technique.
%

\bibliography{ref}
\end{document}